\documentclass[a4paper,12pt]{article}
\usepackage[utf8]{inputenc}
\usepackage{amsthm,amsmath,amsfonts,amssymb,cases}
\usepackage{amsbsy}
\usepackage[colorlinks]{hyperref}
\usepackage{mathtools}
\usepackage{bookmark}
\usepackage[shortlabels]{enumitem}

\usepackage{color}

\usepackage{dsfont}

\usepackage[sort,nocompress]{cite}

\newcommand{\R}{{\mathord{\mathbb R}}}
\newcommand{\Z}{{\mathord{\mathbb Z}}}
\newcommand{\N}{{\mathord{\mathbb N}}}
\newcommand{\C}{{\mathord{\mathbb C}}}

\newcommand{\KK}{\mathcal{K}}
\newcommand{\HH}{\mathcal{H}}

\newcommand{\FF}{\mathcal{F}}

\newcommand{\UU}{\mathcal{U}}
\newcommand{\TT}{\mathcal{T}}
\newcommand{\PP}{\mathcal{P}}
\newcommand{\hh}{\mathfrak{h}}
\newcommand{\vv}{\mathfrak{v}}
\newcommand{\gggg}{\mathfrak{g}}
\newcommand{\rr}{\mathfrak{r}}

\def\one{{\sf 1}\mkern-5.0mu{\rm I}}

\newcommand{\ben}{\begin{displaymath}}
\newcommand{\een}{\end{displaymath}}
\newcommand{\beqn}{\begin{equation}}
\newcommand{\eeqn}{\end{equation}}
\newcommand{\beqna}{\begin{eqnarray*}}
\newcommand{\eeqna}{\end{eqnarray*}}

\newcommand{\inn}[1]{\langle {#1} \rangle }

\newcommand{\DEL}[1]{}

\newtheorem{lemma}{Lemma}
\newtheorem{theorem}[lemma]{Theorem}
\newtheorem{remark}[lemma]{Remark}
\newtheorem{proposition}[lemma]{Proposition}
\newtheorem{corollary}[lemma]{Corollary}
\newtheorem{definition}[lemma]{Definition}

\newtheorem{hypothesis}{Hypothesis}

\numberwithin{equation}{section}

\numberwithin{lemma}{section}

\begin{document}
\title{Symmetries in non-relativistic quantum electrodynamics}
\author{\vspace{5pt} David Hasler$^1$\footnote{
E-mail: david.hasler@uni-jena.de} and Markus Lange$^2$\footnote{E-mail: markus.lange@dlr.de}
\\
\vspace{-4pt} \small{$1.$ Department of Mathematics,
Friedrich Schiller  University  Jena} \\ \small{Jena, Germany }\\
\vspace{-4pt}
\small{$2.$ German Aerospace Center (DLR), Institute for AI-Safety and Security} \\
\small{Sankt Augustin \& Ulm, Germany}\\
}
\date{\today}
\maketitle

\begin{abstract}
We define symmetries in non-relativistic quantum electrodynamics, which
have the physical interpretation of rotation, parity and time reversal symmetry.
We collect transformation properties related to these symmetries in  Fock space  representation as well as  in the
 Schr\"odinger representation.
As an application, we generalize and improve  theorems about Kramer's degeneracy in  non-relativistic quantum electrodynamics.
\end{abstract}

\section{Introduction}

Symmetries are often used to  analyze various properties of physical systems.
In particular in quantum mechanics  symmetries are  used to determine spectral
properties of  the Hamiltonian.
In this paper we  study symmetries of non-relativistic quantum electrodynamics (qed), which
have the physical interpretation of rotation, parity and time reversal symmetry.
We give explicit formulas  for these symmetries both in Fock space representation as well as
in the so called Schr\"odinger representation and apply these symmetries to prove multiplicities
of eigenvalues.

The transformation properties described  in the present paper are of general interest in non-relativistic
qed.  In particular,  in the Fock representation  these
 symmetries  are helpful   for   operator theoretic renormalization analysis of
non-relativistic qed.  On the one hand,  symmetries can be used
to control marginal terms  \cite{Sig09,HasHer11-1,HasHer11-2,HaslerHerbst.2012-4}. On the other hand symmetries allow
  the treatment of  degenerate
eigenvalues in the frame work of  renormalization, provided  the symmetries act irreducibly
on the eigenspace \cite{HasLan22}.  In fact, the latter   is the main interest, which we had  in mind, for collecting
the transformation properties  of the aforementioned symmetries.

In physics literature  continuous symmetries  are often described by means of their  infinitesimal generator.
That is, as a representation of the Lie-algebra.
For non-relativistic qed the  generators of the Lie-algebra of $SU(2)$  are readily available in textbooks about non-relativistic
aspects of quantum electrodynamics
\cite{CohenGrynbergDupont.1992,Spo04}. In this paper we express the $SU(2)$-symmetry
directly  as a  representation of  the  Lie-group.

As already mentioned, symmetries are  helpful  in the spectral analysis of Hamiltonian operators of quantum mechanics.
For example  the classical Kramers degeneracy theorem states, that the eigenvalues of a  time-reversal  symmetric  Hamiltonian describing an odd number of spin 1/2-particles have even multiplicity.
Using  a  theorem of this type  it was shown  in  \cite{LosMiySpo09,LosMiySpo10}  that  Hamiltonians of non-relativistic qed, which describe
 odd number of spin 1/2-particles have a doubly degenerate ground state, provided the external potential is symmetric
with respect to parity.
In this paper we improve that result and show  that   parity symmetry is not necessary.
This is of physical relevance,  since potentials describing  molecules with static nuclei,  are not necessarily
symmetric with respect to parity.
Furthermore, we include external magnetic fields in the  mathematical model.
Finally, we consider translation invariant systems and generalize degeneracy results for  a single
spin 1/2-particle   \cite{HiroshimaSpohn.2001,Hiroshima.2007} to atoms and molecules.

Let us give a short outline of the paper. In the next  section we
review the notion of a symmetry in quantum mechanics and state an abstract version
of Kramers degeneracy theorem. In Section \ref{sec:model}
 we introduce non-relativistic qed. In Section \ref{sec:sym}
we  define rotation, parity, and  time-reversal  symmetry. Moreover we collect various transformation
properties.
In Section \ref{sec:ham}  we study symmetry
properties of Hamiltonians of non-relativistic qed. In particular we  show the aforementioned degeneracy  theorems.
 In Section  \ref{sec:trainv} we study symmetry
properties of  fibers of translationally invariant Hamiltonians of non-relativistic qed. In Section  \ref{sec:schr} we
define   rotation, parity, and time-reversal symmetry  in the so called Schr\"odinger representation.
We  show that the  definitions in Schr\"odinger representation
 agree with the definitions in Fock space representation. To show this, we use    the  canonical
unitary transformation mapping the Fock space representation to the Schr\"odinger representation.

\section{Symmetries in Quantum Mechanical Systems}

In this section we collect some well-known definitions and properties.

\begin{definition} Let $V$ be a complex vector space. A mapping $A : V \to V$ is called
{\bf anti-linear} (or {\bf conjugate linear})  if
\begin{itemize}
\item[(i)] $A (x + y ) = A x + A y$ for all $x,y \in V$.
\item[(ii)] $ A (\alpha x ) = \overline{\alpha} A x $, for all $x \in V$ and $\alpha \in \C$.
\end{itemize}
If  $\HH$ is a complex Hilbert space and  $T : \HH \to \HH$ anti-linear, then
 the {\bf adjoint}  $T^* : \HH \to \HH$ is defined by
$$
\inn{ T^* x  , y  } = \inn{   T y , x  }   \quad , \quad \forall  x , y \in \HH .
$$
If  $\HH$ is  a complex Hilbert space and
 $S : \HH \to \HH$ anti-linear, then $S$ is called  {\bf anti-unitary}  if it is surjective and satisfies
$$
\inn{ S x , S y }  =  \inn{y , x }
 \quad , \quad \forall  x , y \in \HH.
$$
\end{definition}

The assertions of the following Lemma are straightforward to verify.

\begin{lemma} \label{lem:antiline}  The following holds.
\begin{itemize}
\item[(a)] Let $C_i$ be anti-linear (anti-unitary) transformations on complex vector spaces $V_i$ (Hilbert spaces), $i=1,2$. Then $C_1 \otimes C_2 : V_1 \otimes V_2 \to V_1 \otimes V_2$ is
also anti-linear (anti-unitary).
\item[(b)] If $T : \HH \to \HH$ is an  anti-linear mapping on a Hilbert space $\HH$, then also $T^*$
is anti-linear.
\item[(c)]
If $S$ is anti-unitary, then $S$ is bijective and  $S^* S = 1$ and $SS^* = 1$.
\end{itemize}
\end{lemma}

\begin{definition}
Let $S$ be a  unitary or  anti-unitary operator.
Let $H$ be a densely defined  operator in $\HH$.  We call  $S$   a {\bf symmetry of} $H$,   if
$$
S H =  H S
$$
when  $S$ is unitary,
and
$$
S H =  H^* S
$$
when $S$ is anti-unitary.
\end{definition}

 The following theorem, whose  formulation is from \cite{LosMiySpo09}, can be viewed as an  abstract  version of Kramer's degeneracy theorem, \cite{Kra30,Wigner1932}.

\begin{theorem}[Abstract Kramers Degeneracy] \label{le:kram} Let $\theta$ be a an anti-unitary symmetry of a self-adjoint operator $H$ and $\theta^2 = - 1$. Then each eigenvalue of $H$ is at least
doubly degenerate. Any eigenvalue of $H$ with finite multiplicity has  even multiplicity. 
\end{theorem}

The proof follows from the following lemma.

 \begin{lemma}\label{lem:timeinvirred} Let $J$ be an anti-unitary operator on a complex Hilbert space $V$ with $J^2  = -1$.
Then the following holds.
\begin{itemize}
\item[(a)]   For any nonzero $v \in V$,  also  $Jv$ is nonzero and   $v \perp J v$.
\item[(b)]   The  Hilbert space  $V$ cannot have finite odd dimension.
\end{itemize}
\end{lemma}
\begin{proof}
(a) Since $J(Jv) = J^2 v = - v$, the vector $Jv$ is nonzero.
  Since $J$ is anti-unitary
$$
\inn{ v , J v } = \inn{ JJ v , J v } = - \inn{ v , J v} .
$$
So $\inn{ v , Jv } =0$. \\
(b)  We show by induction   that  $V$ cannot have dimension $2n-1$ for $n \in \N$.  Clearly, the induction hypothesis  holds true for $n=1$ by (a). Suppose the induction hypothesis holds for $n$, and
suppose $V$ has dimension $2n+1$.  Pick a nonzero $v \in V$. Then  $J v  \in V$ and $J v \perp v$ by (a).
Thus    $W := \{v , J v \}^\perp$  is a  complex  vector  space,   which has  dimension $2n-1$. Since $J^2 = -1$, it follows that $J$ leaves the complex linear span ${\rm lin}_\C \{ v , J v \}$ invariant.
Since  $J$ is anti-unitary, it  leaves also $W$ invariant. But the complex vector space $W$  together with  $J|_W$  contradict the induction hypothesis.
\end{proof}

\begin{proof}[Proof of Theorem \ref{le:kram}] Let $E$ be an eigenvalue of $H$. Since $H$ is self-adjoint $E$ is real.  So  $\theta$   leaves the space $V = {\rm ker}( H-E)$ invariant, since
$ ( H - E ) \theta \psi = \theta ( H - E ) \psi .  $
Thus the first and second  statement follow from  (b) of Lemma \ref{lem:timeinvirred}  with   $J=\theta$.
\end{proof}

\section{Non-relativistic qed} \label{sec:model}

For a complex Hilbert space $\HH$ we  denote   
 the $n$-fold tensor product  by
$$
\HH^{\otimes n} := \bigotimes_{j=1}^n \HH  
$$
and we set $\HH^{\otimes 0} := \C$.
Let  $\mathfrak{S}_{\{1,...,n\}}$ be the permutation group of the set $\{1,...,n\}$. For each $\sigma \in \mathfrak{S}_{\{1,...,n\}}$ we define an operator  $\mathfrak{U}(\sigma)$
on $\HH^{\otimes n}$ by
\begin{align} \label{defofsymmop}
\mathfrak{U}(\sigma) (\varphi_1 \otimes \varphi_2 \otimes \cdots \otimes  \varphi_n) = \varphi_{\sigma(1)} \otimes \varphi_{\sigma(2)} \otimes \cdots \otimes \varphi_{\sigma(n)}
\end{align}
for any $\varphi_j \in \HH$, $j=1,...,n$,  and extending it linearly. This yields a bounded operator  (of norm one) on  $\HH^{\otimes n}$ so we
can define $S_n = \frac{1}{n!}\sum_{\sigma \in \mathfrak{S}_{\{1,...,n\}}} \mathfrak{U}( \sigma )$.
We define the symmetric  $n$-fold tensor product of $\HH$ by
\begin{align*}
\HH^{\otimes_s n } := {S}_n \left( \HH^{\otimes n } \right) .
\end{align*}
Let  $\mathcal{D}_s$ denote the representation space of $SU(2)$ with dimension $2s+1$.  In this paper we  shall only  consider
the case  $s=0$, describing spinless particles,  and the case   $s  =  \frac{1}{2}$, describing particles with spin $1/2$.

The model consists of
$N$ particles   with spins $s_j \in \{0,1/2\} $,
masses $m_j > 0$, charges $q_j \in \R$,   values of the spin magnetic moments   $\mu_j \in \R$,   $j=1,...,N$. By $x_j \in \R^3$ we shall
denote the position of the $j$-the particle.
The Hilbert space describing the non-relativistic quantum mechanical matter  is
$$
\HH_{{\rm mat}} =  \bigotimes_{j=1}^N  L^2(\R^{3};\mathcal{D}_{s_j}) .
$$
We note that the description  of physical systems  usually requires the   restriction to  a subspace  determined by the  particle statistics of  identical  particles.  This will be considered below.

If  $s=0$,  let    $\hat{s}_l = 0$ for $l=1,2,3$, and   if  $s=1/2$,  let   $\hat{s}_l = \frac{1}{2} \sigma_l$ for $l=1,2,3$, where $\sigma_l$
denotes the $l$-th  Pauli-matrix.
\begin{remark} \label{paulitau}  {\rm  Note that  $\hat{s}_1$, $\hat{s}_2$ and $\hat{s}_3$ are representations
of the  generators of $su(2)$ in the representation $\mathcal{D}_s = \C^{2s+1}$, $s \in \{ 0,1/2\}$. They are  linear maps in $\mathcal{D}_s$
 satisfying
\begin{align}&
[  \hat{s}_j ,\hat{s}_k ]  = \sum_{l=1}^3  i \epsilon_{j,k,l} \hat{s}_l ,  \quad \hat{s}_l^* = \hat{s}_l , \ l=1,2,3,  \nonumber  \\
& \quad \overline{\hat{s}}_1 = \hat{s}_1 , \quad \overline{\hat{s}}_2 = - \hat{s}_2 , \quad \overline{\hat{s}}_3 = \hat{s}_3 , \label{compspin}
\end{align}
where $\epsilon_{j,k,l}$ denotes the totally antisymmetric tensor in three dimensions.
}
\end{remark}
For $j=1,...,N$ and $l=1,2,3$ we define
$$
(\widehat{S}_j)_l   =  \left( \bigotimes_{k=1}^{j-1} \one_{\mathcal{D}_{s_k}}  \right) \otimes \hat{s}_l \otimes  \left( \bigotimes_{k=j+1}^{N} \one_{\mathcal{D}_{s_k}}   \right) .
$$
For a Hilbert space $\hh$  define the symmetric
 Fock space over $\hh$  by
$$
\FF_s(\mathfrak{h}) := \bigoplus_{n=0}^\infty  \hh^{\otimes_s n } . 
$$
Thus  we can  identify $\psi \in \FF_s(\hh)$ with a sequence
of functions $\psi = (\psi^{(0)}, \psi^{(1)},\psi^{(2)},....)$ such that $\psi^{(n)} \in \hh^{\otimes_s n }$. 
We introduce  the set  $F_0(\hh) := \{ \psi \in \FF_s(\hh) : \exists N , \forall n \geq N ,  \psi^{(n)} = 0 \}$
of finite particle vectors.
For $f \in \mathfrak{h}$ let $a^*(f)$  denote the usual creation  operator, which is a
densely defined closed linear operator which satisfies  for $\eta \in \hh^{\otimes_s n }$ 
\begin{align} \label{defastar}
a^*(f) \eta =  \sqrt{n+1} S_{n+1}   ( f \otimes \eta  ) .
\end{align}
Let $a(f)$ denote the adjoint of the creation operator.
If  $T$ be a symmetry  in $\mathfrak{h}$,  then $\Gamma(T)$ denotes  the unique  operator on
 $\FF(\mathfrak{h})$ such that on  $\mathfrak{h}^{\otimes_s n }$
$$
\Gamma(T) |_{  \mathfrak{h}^{\otimes_s n }   } = \bigotimes_{j=1}^n T . 
$$
It is straight forward to see that also $\Gamma(T)$ is a symmetry.
Let $A$ be any self-adjoint operator on $\HH$ with domain of essential self-adjointness $D$.
Let $D_A = \{ \psi \in F_0(\hh)  :   \psi^{(n)} \in \otimes_{k=1}^n D \text{ for each } n \}$ and
define $d \Gamma(A)$ on $D_A \cap \mathfrak{h}^{\otimes_s n } $ as
$$
A \otimes 1 \otimes \cdots \otimes 1  + 1 \otimes A \otimes \cdots \otimes 1  + \cdots + 1 \otimes \cdots \otimes 1 \otimes A .
$$
In \cite[Section VIII.10]{reesim1} it is shown that $d \Gamma(A)$ is essentially self-adjoint on $D_A$, and
we shall denote this self-adjoint extension again  by $d\Gamma(A)$.
It follows from the definitions that for a symmetry $T$ on $\hh$ and $f \in \hh$
\begin{align}
& \Gamma(T) a^\#(f) \Gamma(T)^* = a^\#(T f), \label{eq:GammaasharpGamma} \\
&  \Gamma(T ) d \Gamma(A)  \Gamma(T)^* = d \Gamma( T A T^*) \label{eq:GammdGammaGamma},
\end{align}
where $a^\#$ stands for $a$  or  $a^*$.
Let us now define operators acting on the composite Hilbert space
 $$
 \HH  \otimes \FF_s(\hh)   ,
 $$
where $\HH$  denotes a  Hilbert space, which is  used to describe the matter.
For  a bounded linear operator
$G  \in \mathcal{L}(\HH , \HH \otimes \hh )$   we define
 for $\varphi \in \HH$ and  $\eta \in  \mathfrak{h}^{\otimes_s n }$ 
\begin{align} \label{defofG}
a^*(G)  ( \varphi \otimes \eta )  = \sqrt{n+1}  ( 1 \otimes S_{n+1} ) (  ( G \varphi ) \otimes \eta )  .
\end{align}
This extends by linearity to a closable operator   in $\HH \otimes \FF_s(\hh)$, which we shall
again denote by $a^*(G)$. We define  $a(G) = [a^*(G)]^*$.

In non-relativistic qed one consider the Fock space over $\gggg := L^2(\R^3 \times \Z_2)$. In that case we can identify $\psi \in \FF_s(\gggg)$ with a sequence
of functions $\psi = (\psi^{(0)}, \psi^{(1)},\psi^{(2)},\ldots)$ such that $\psi^{(n)} \in L^2_s((\R^3 \times \Z_2)^n)$, where the subscript $s$ stands for wave functions which are
symmetric with respect to interchange of components of the $n$-fold Cartesian product.
Let $M_f$ denote  the operator of multiplication by the function $f$.
We define
$$
H_{\rm f} = d \Gamma(M_\omega)  ,
$$
where  the so-called dispersion relation $\omega : \R^3 \to [0,\infty)$ is defined such that $\omega(k) = \omega(k')$ whenever $|k'| = |k|$. 
Moreover define
$$
P_{\rm f} = d \Gamma( M_{\pi_j } ) ,
$$
where $\pi_j : \R^3 \to \R$ with $\pi_j(k) = k_j$.  Next we introduce   creation and annihilation  operators  in terms of  operator  valued distributions.
We define
$$
\mathcal{D}_\mathcal{S} := \{ \psi \in  F_0(\gggg) : \ \psi^{(n)} \in \mathcal{S}((  \R^{3} \times \Z_2   )^n )  \} .
$$
where $\mathcal{S}(( \R^{3} \times \Z_2 )^n )$ denotes the space of smooth  rapidly decaying functions.
For each $(k,\lambda) \in \R^3 \times \Z_2$ we define an operator $a(k,\lambda)$ on $\FF_s(\gggg)$ with domain $\mathcal{D}_\mathcal{S}$ by
$$
(a(k, \lambda) \psi)_n(k_1,\lambda_1,...,k_n,\lambda_n) = \sqrt{n+1} \psi_{n+1}(k,\lambda,k_1,\lambda_1,...,k_n,\lambda_n) .
$$
We define $a^*(k,\lambda)$ in the sense of quadratic forms on $\mathcal{D}_\mathcal{S} \times \mathcal{D}_\mathcal{S}$ by
$$
\inn{\psi_1, a^*(k,\lambda) \psi_2 } = \inn{ a(k,\lambda)\psi_1, \psi_2 } .
$$
Then it is straight forward to see  that
$$
a^*(f) = \sum_{\lambda=1,2} \int_{\R^3}  f(k,\lambda) a^*(k,\lambda) dk , \quad
a(f) = \sum_{\lambda=1,2} \int_{\R^3}  \overline{f(k,\lambda) } a(k,\lambda) dk ,
$$
where the equalities are understood in the sense of quadratic forms and the integrals are
understood as weak integrals.
Let us now relate the definition given in  \eqref{defofG} to  integrals of  operator valued distributions.
To this end we use  the natural embedding
\begin{align*}
I :  L^2(\R^3 \times \Z_2  ; \mathcal{L}(\HH_{{\rm mat}}))  & \to \mathcal{L}(\HH_{{\rm mat}} ;  L^2(\R^3 \times \Z_2 ; \HH_{\rm mat}  )    )  \cong \mathcal{L}(\HH_{{\rm mat}} ; \HH_{\rm mat} \otimes \gggg) \\
 g & \mapsto \left( \varphi \mapsto  [ (k,\lambda) \mapsto  g(k,\lambda)  \varphi ]  \right)  ,
\end{align*}
 which is a bounded injection, cf. \cite[Theorem II.10]{reesim1}.
Then for $g \in L^2(\R^3 \times \Z_2  ; \mathcal{L}(\HH_{{\rm mat}})) $
it is straight forward to show that
\begin{align} \label{creadiffdef}
a^*( I(g))  = \sum_{\lambda=1,2}  \int_{\R^3}  g(k) \otimes  a^*(k,\lambda) d k  , \quad a( I(g))  = \sum_{\lambda=1,2} \int_{\R^3}  g(k)^* \otimes  a(k,\lambda)   dk
\end{align}
  in the sense of quadratic forms on $\HH_{\rm mat}  \otimes \mathcal{D}_\mathcal{S} $, where the integral is a  weak integral. Henceforth
we shall drop the tensor sign in  \eqref{creadiffdef}   if it is clear on which factor the operator acts.
The definition of the vector potential involves
the so called polarization vectors.
For $\lambda=1,2$ we choose a  measurable function
\begin{equation}  \label{eq:epsilon}
\varepsilon(\cdot , \lambda) : {\rm  S}_2 \to \R^3
\end{equation}
on the 3-dimensional sphere ${\rm S}_2$ with the following properties.   For each $k \in {\rm S}_2$ the vectors $( \varepsilon(k,1), \varepsilon(k,2),k )$ form an orthonormal basis of $\R^3$.
We extend $\varepsilon(\cdot , \lambda)$ to $\R^3 \setminus \{ 0 \}$ by setting
$
\varepsilon(k,\lambda) := \varepsilon(k/|k|,\lambda)
$
for all nonzero $k$. We  assume that we are given a  measurable coupling function $\kappa : \R^3 \to \C$.
 We note that the Fourier transform of $\kappa$ is real, if and only if
\begin{align} \label{eq:assonkappa}
\overline{\kappa(k)} =  \kappa(-k)  .
\end{align}
We define the coupling functions for $l=1,2,3$ and $x  \in \R^3$
$$
g_{x,l}^{(\varepsilon)}(k,\lambda)  = \frac{ \varepsilon_l(k,\lambda) }{\sqrt{2 \omega(k) }}\overline{\kappa(k)} { e}^{ - i k \cdot x }  .
$$
We can now define the field operators. If $\omega^{-1/2} \kappa \in L^2(\R^3)$, we define the magnetic vector potential
\begin{align*}
A_l(x)  := &   a(g^{(\varepsilon)}_{x,l}) +  a^*(g^{(\varepsilon)}_{x,l}) \\
  = & \sum_{\lambda=1,2} \int_{\R^3}   \frac{\varepsilon_l(k,\lambda)  }{\sqrt{2 \omega(k) }} \left( \kappa(k) e^{  i k \cdot x }  a(k,\lambda) +  \overline{\kappa(k)} e^{ - i k \cdot  x }  a^*(k,\lambda)  \right)  dk  , \quad l=1,2,3 ,
\end{align*}
where in the second line we made use of  \eqref{creadiffdef}.
If $| \cdot| \omega^{-1/2} \kappa \in  L^2(\R^3)$, we define the quantized magnetic field
\begin{align*}
B_l(x)  := &  [\nabla \times A(x)]_l   \\
  = &\sum_{\lambda=1,2} \int_{\R^3}   \frac{ i [ k \times \varepsilon(k,\lambda) ]_l }{\sqrt{2 \omega(k) }} \left( \kappa(k) e^{  i k \cdot x }  a(k,\lambda) -   \overline{\kappa(k)} e^{ - i k \cdot x }  a^*(k,\lambda)  \right)  dk  ,  \quad l=1,2,3 .
\end{align*}
If $\omega^{1/2} \kappa \in L^2(\R^3)$, we define the quantized electric field
\begin{align*}
E^\perp_l(x)  := &   a(- i \omega g^{(\varepsilon)}_{x,j}) +  a^*(-i \omega g^{(\varepsilon)}_{x,j}) \\
  = & \sum_{\lambda=1,2} \int_{\R^3}    i  \varepsilon_l(k,\lambda) \sqrt{\frac{\omega(k)}{2}}    \left( \kappa(k) e^{  i k
\cdot  x }  a(k,\lambda) -   \overline{\kappa(k)} e^{ - i k \cdot x }  a^*(k,\lambda)  \right)  dk  , \quad l =1,2,3 .
\end{align*}
The Hamiltonian acting in the Hilbert space
\begin{align*}
 \HH_{\rm mat} \otimes \FF_s(\mathfrak{g} )
\end{align*}
is given by
\begin{align}  \label{defofham}
H  &  = \sum_{j=1}^N  \left\{ \left(  p_j \otimes 1  +  q_j ( A(\hat{x}_j)   + A_{\rm ext}(\hat{x}_j) ) \right)^2   + \mu_j  \widehat{S}_j  \cdot  ( B(\hat{x}_j) + B_{\rm ext}(\hat{x}_j)) \right\} \nonumber  \\
& \quad + 1 \otimes  H_{\rm f}  + V(\hat{x}_1,...,\hat{x}_N) \otimes 1   ,
\end{align}
where $\hat{x}_j$  denotes the operator of multiplication with $x_j$, the coordinates of the  $j$-th particle,  and   $p_j = - i \nabla_{x_j}$. 
We assume that  $V : \R^{ 3 N} \to \R$  is a function and that $B_{\rm ext} : \R^3 \to \R^3$ is a function.  Furthermore, we   defined
\begin{equation} \label{exprforAext}
A_{{\rm ext}}(x) :=  -  \int \frac{(x-y) \times B_{\rm ext}(y)}{4\pi |x-y|^3} dy ,
\end{equation}
cf. Remark \ref{rembextaest}. 

\begin{remark} \label{rembextaest}  {\rm
Provided $B_{\rm ext}$ is sufficiently regular and has sufficient decay, we can write 
\begin{equation}  \label{eq:defofA}
A_{{\rm ext}}(x) :=  -  \int \frac{(x-y) \times B_{\rm ext}(y)}{4\pi |x-y|^3} dy =  \nabla_x \times \int \frac{ B_{\rm ext}(y)}{4\pi |x-y|} dy
 =   \int \frac{\nabla_y \times  B_{\rm ext}(y)}{4\pi |x-y|} dy ,
\end{equation}
by calculating the derivative and using integration by parts, respectively.   In particular,  if $\nabla \cdot B_{\rm ext} = 0$, it follows that  $\nabla \times A_{\rm ext} = B_{\rm ext}$.

}
\end{remark}

Physically, $V$ is called the external  potential, $B_{\rm ext}$ the external magnetic field, $A_{\rm ext}$ the  external magnetic vector potential. 
We assume that  $B_{\rm ext}$ is  such that  $A_{\rm ext}$  in  \eqref{exprforAext}  is  well  defined for almost all $x \in \R^3$.
Moreover, we assume that   $\kappa$ and $\omega$ are  such that the  fields occurring
in the Hamiltonian exist. Furthermore,  we assume that $\kappa$, $\omega$,  $V$, and $B_{\rm ext}$ are such  that the  Hamiltonian is essentially self-adjoint on  $\left( \bigotimes_{j=1}^N C_c^\infty(\R^3; \mathcal{D}_{s_j} )  \right) \otimes  F_0(\gggg)$, for details we refer the reader to  \cite[Theorem X.35, Theorem X.34]{ReeSim2} and \cite{HasHer08-2}.

\section{Symmetries} \label{sec:sym}

In this subsection we  define symmetries associated to rotations, space inversion and
time inversion. To define these symmetries on Fock space it is convenient to identity
$\hh$ with the space of so called divergence free  vector fields. In this section
we shall denote by $F$   the Fourier  transform and
by $F^{-1}$ its inverse, i.e. for $f \in L^1(\R^3)$
\begin{align*}
(F f)(k) & =  (2\pi)^{-3/2}  \int_{\R^3}  e^{ - i k \cdot x } f(x) dx   , \\
(F^{-1}  f)(x) & =   (2\pi)^{-3/2} \int_{\R^3}  e^{  i k \cdot x } f(k) dk   ,
\end{align*}
where both transformations are canonically  extended to $L^2(\R^3)$  by Plancherel's theorem.

\subsection{Space of divergent free vector fields}

We introduce the space of divergence free vector fields
$$
\mathfrak{v} := \{ v \in L^2(\R^3 ; \C^3) :   \sum_{j=1}^3 k_j  \widehat{v}_j(k) = 0 \ , \ \text{a.e.} \ k  \in \R^3  \} .
$$
Given a specific measurable choice for the polarization
vectors \eqref{eq:epsilon}  we obtain a canonical identification with the one photon
Hilbert  space
$\gggg = L^2(\R^3 \times \Z_2)$.
This is the content of the following lemma.

\begin{lemma}  \label{lem:trafohtov}  For  the polarization vector  $\varepsilon : S_2 \times \Z_2  \to \R^3$,   as in  \eqref{eq:epsilon},
the map
$$
\tau_\varepsilon : \gggg \to \mathfrak{v} , \quad h \mapsto  \left( F^{-1} \sum_{\lambda=1,2} \varepsilon_j( \cdot , \lambda ) h(\cdot,\lambda) \right)_{j=1,2,3}  ,
$$
is unitary and its  inverse acting on $v \in \vv$ is determined  by
$
(\tau_\varepsilon^{-1} v )(k,\lambda) = \varepsilon(k, \lambda)  \cdot ( F v)(k)
$ for almost all $(k,\lambda) \in \R^3 \times \{1,2\}$.
\end{lemma}

For the proof we first note the following. For $ k \in \R^3 \setminus \{ 0 \}$  define
\begin{align}  \label{defofP}
P(k)_{a,b} := \delta_{ab} - \frac{k_a k_b}{|k|^2}  , \quad a , b = 1,2,3 , \quad k \neq 0
\end{align}
From the definition it follows that $P(k)_{a,b} = P(k)_{b,a}$, and that  $P(k)$ is equal to  the projection operator in $\C^3$ onto the subspace in $\C^3$, which is perpendicular to $k$.
Thus  from the definition of the polarization vectors,  \eqref{eq:epsilon}, we infer that for $k \in \R^3 \setminus \{ 0 \}$
\begin{align} \label{projeident}
P(k)_{a,b}  =  \sum_{\lambda=1,2} \varepsilon_a(k,\lambda) \varepsilon_b(k,\lambda) .
\end{align}

\begin{proof}[Proof of Lemma   \ref{lem:trafohtov}]  The lemma  follows from a straight
forward calculation using the properties of the polarization vectors.
Let $h \in \gggg$.
Clearly, $\tau_\varepsilon$ is well defined, since $k \cdot F(\tau_\varepsilon(h))(k) = k \cdot \sum_{\lambda=1,2}  \varepsilon(k,\lambda) h(k,\lambda) = 0$.
The map is an isometry, since
\begin{align*}
\| \tau_\varepsilon h \|^2  & =  \int_{\R^3}   \sum_{j=1}^3  \sum_{\lambda,\lambda'=1,2} \overline{ \varepsilon_j(k,\lambda) h(k,\lambda) }  \varepsilon_j(k,\lambda') h(k,\lambda')d^3 k \\
& = \int_{\R^3}    \sum_{\lambda,\lambda'=1,2} \delta_{\lambda,\lambda'}  \overline{ h(k,\lambda) }   h(k,\lambda')  d^3 k= \| h \|^2 .
\end{align*}
Furthermore for $v \in \mathfrak{v}$  let $(\beta_\epsilon v)(k,\lambda) = \varepsilon(k,\lambda) \cdot (  Fv)(k)$.
Then
\begin{align*}
F (\tau_\varepsilon (\beta_\varepsilon v)_j)(k)    & =  \sum_{\lambda=1,2} \varepsilon_j( k , \lambda ) (\beta_\varepsilon v)( k ,\lambda)  \\
& =  \sum_{\lambda=1,2}  \varepsilon_j( k , \lambda ) \sum_{l=1}^3  \varepsilon_l(k,\lambda) \cdot (  Fv_l)(k) \\
& = Fv_j(k)  ,
\end{align*}
where we used that \eqref{projeident}  and that $v$ is divergence free. This shows the surjectivity of $\tau_\varepsilon$
and that its inverse is given by  $\beta_\varepsilon$.
\end{proof}

Define for $x\in \R^3$ and $a=1,2,3$ the function $v_{x,a} : \R^3 \to \C^3$ by
\begin{align} \label{defofvI}
[v_{x,b}(y)]_{a} := \frac{1}{(2\pi)^{3/2}} \int_{\R^3}  e^{- i k \cdot  (x-y)} \frac{\overline{\kappa(k)}}{\sqrt{2 \omega(k) }} P(k)_{a,b}\, dk , \quad y \in \R^3 .
\end{align} 
The properties collected in the following lemma are straight forward to verify using the definitions.

\begin{lemma}\label{lem:reltinvpol-1} We have the following properties for $x \in \R^3$ and $b =1,2,3$
\begin{enumerate}[(a)]
\item  $v_{x,b} \in \mathfrak{v}$,
\item   $ v_{x,b} = \tau_\varepsilon g^{(\varepsilon)}_{x,b}$, \quad $\tau_\varepsilon^{-1} v_{x,b} = g^{(\varepsilon)}_{x,b}$. 
\end{enumerate}
\end{lemma}

The  next    lemma  will be needed to determine  transformation properties of the field energy and field momentum
 with respect to rotations,  parity transformations, and time reversal symmetry.

\begin{lemma}\label{lem:reltinvpol00}  Let  $f :  \R^3 \to \C$ be  a measurable function. 
 Then we have the following properties.
\begin{enumerate}[(a)]
\item  $\tau_\varepsilon M_f \tau_\varepsilon^{-1} = F^{-1} M_f F $.
\item   \label{eq:derink000}  For $\varphi \in L^2(\R^3)$  and $S \in O(3)$ we define the transformation  $T_S \varphi = \varphi \circ S^{-1} $. Then
\begin{align}
T_S  F    & =    F T_S ,   \quad T_S  F ^{-1}     =    F^{-1}  T_S    \label{eq:derink0} .
\end{align}
\item Let $T_S$ be defined as in \ref{eq:derink000}. Then  $T_S^{-1} = T_{S^{-1}}$ and
\begin{equation}  \label{eq:derink} T_S M_f T_S^{-1} = M_{f \circ S^{-1}} .
\end{equation}
\end{enumerate}
\end{lemma}
\begin{proof} Part (a) follows from
\begin{align}
(\tau_\varepsilon    M_f  \tau_\varepsilon^{-1}  v )_j  & =    F^{-1} \sum_{\lambda=1,2} \varepsilon_j( \cdot , \lambda )
M_f  (\tau_\varepsilon^{-1}  v)(\cdot,\lambda)  \nonumber   \\
& =  F^{-1} \sum_{\lambda=1,2} \varepsilon_j( \cdot , \lambda )   f(\cdot)    \varepsilon(\cdot, \lambda) \cdot (F  v)(\cdot)  \nonumber \\
& = F^{-1}  ( M_f   F  v_j )\label{eq:derink-1}
\end{align}
(b)
If  $\varphi \in \mathcal{S}(\R^3)$, we find  by  the transformation formula for integrals  for arbitrary $S \in O(3)$
\begin{align}
(T_S  F \varphi)(  k) & =  (2\pi)^{-3/2}  \int_{\R^3}   e^{-i (S^{-1} k)  \cdot  x  }   \varphi(x) dx =   (2\pi)^{-3/2}  \int_{\R^3}   e^{-i k  \cdot x  } \varphi(S^{-1} x) dx = (F T_S \varphi)(k)  \label{eq:derink00} .
\end{align}
So (b) follows  by density and continuity.
Part (c) is straight forward to verify.
\end{proof}

The following   lemma  will be needed to determine  transformation properties of the interaction with respect to rotations and parity transformations.

\begin{lemma}\label{lem:reltinvpol}  Let $S \in O(3)$.  Then the following holds.
\begin{enumerate}[(a)]
\item  For all $k \in \R^3$ we have  $P(Sk) = S  P(k) S^{T}$.
\item \label{lem:reltinvpolb}  For all $x \in \R^3$  and $b =1,2,3$
\begin{align}
 &  \sum_{c'=1}^3 S_{c,c'} \int_{\R^3} e^{- i k \cdot (x- S^{-1} y)} \frac{\overline{\kappa(k)}}{\sqrt{ 2 \omega(k)} } P_{b,c'}(k)    dk \label{eq:trafoforho-1}  \\
& \quad    = \sum_{b' =1}^3 S_{b',b} \int_{\R^3}  e^{- i  k \cdot (S x- y)}\frac{\overline{\kappa( S^{-1} k)}}{\sqrt{ 2 \omega( k)} } P_{b',c}( k)     dk  . \label{eq:trafoforho0}
\end{align}
\item lf $\kappa(S \cdot ) = \kappa(\cdot)$, then for all  $x \in \R^3$  and $b =1,2,3$
\begin{align}\label{eq:trafoforho05}
Sv_{x,b}(S^{-1} y) = \sum_{b'=1}^3 S_{b',b} v_{Sx,b'}(y)  .
\end{align}
\end{enumerate}
\end{lemma}

\begin{proof} Part  (a) is straight forward to verify using the definition   \eqref{defofP}. For $x \in \C^3$  and  $k \in \R^3 \setminus \{ 0 \}$
we find  for $\widehat{k}  = k/|k|$
$$
P(Sk) x = x - S  \widehat{k}   (  S \widehat{k}  \cdot   x  )    = S S^T x -  S  \widehat{k}  (  \widehat{k} \cdot    S^T x )  =   S P(k) S^T x .
$$
(b) follows from a change of variables and (a) 
\begin{align}
 & \eqref{eq:trafoforho-1} = \sum_{c'=1}^3 \int e^{- i ( S k ) \cdot (S x- y)}  \frac{\overline{\kappa(k)}}{\sqrt{ 2 \omega(k)} }      P_{b,c'}(k)  S_{c,c'}    dk  \nonumber   \\
& \quad = \sum_{c'=1}^3 \int e^{- i  k \cdot (S x- y)}  \frac{\overline{\kappa(  S^{-1} k)}}{\sqrt{ 2 \omega(k)} }  P_{b,c'}( S^{-1} k)    S_{c,c'}    dk =   \eqref{eq:trafoforho0}  .  \nonumber
\end{align}
(c)
Now \eqref{eq:trafoforho05}  follows  from   \eqref{eq:trafoforho0}  and  the definition \eqref{defofvI}.
\end{proof}

\subsection{Rotation Invariance}

 We  introduce
the  so called canonical  double covering  homomorphism
$$
\pi  : SU(2) \to SO(3)  ,  \quad U \mapsto \pi(U) ,
$$
where $\pi(U)$ is the unique element of $SO(3)$ such that
$$
U \sigma_m U^* =  \sum_{l=1}^3 \pi(U)_{l,m} \sigma_l  , \quad m =1,2,3 ,
$$
with  $\sigma_1, \sigma_2, \sigma_3$ denoting  the Pauli matrices.
On the one electron Hilbert space $L^2(\R^3;\mathcal{D}_s)$ we define
$$
(\UU_{{\rm p}, s}(U) \psi)(x) = D_s(U) \psi(\pi(U)^{-1} x) ,
$$
where $D_s$ denotes the  representation of $SU(2)$ with spin $s$. Similarly we define for $v \in \mathfrak{v}$ the transformation for $R \in SO(3)$
$$
(\UU_{\mathfrak{v}}(R) v)(x) = R v(R^{-1} x) .
$$
Moreover,  we define
$$
 \UU_{\gggg}(R) = \tau_\epsilon^{-1}  \UU_{\mathfrak{v}}(R) \tau_\epsilon  ,
$$
which  depends on the choice of the polarization vectors.
For $R \in SO(3)$ we define the unitary mapping
\begin{align*}
\UU_{\rm f}(R) & =  \Gamma(\UU_{\gggg}(R)) ,
\end{align*}
and  for $U \in SU(2)$ we define the unitary mappings
\begin{align*}
\UU_{\rm mat}(U) & = \bigotimes_{j=1}^N \UU_{{\rm p},s_j}(U)  \\
\UU(U) &= \UU_{\rm mat}(U)  \otimes \UU_{\rm f}(\pi(U))
\end{align*}
on the Hilbert spaces $\HH_{\rm mat}$ and
$\HH_{\rm mat} \otimes \FF_s(\mathfrak{g})$, respectively. This defines a representation of $SU(2)$ on these
Hilbert spaces.
The next proposition collects elementary properties, which  follow directly from the definitions.

\begin{proposition} The map $\UU_{\rm f}$ is a unitary representation of $R \in SO(3)$, and the maps   $\UU_{\rm mat}$,
 and  $\UU$ are unitary representations of $SU(2)$.
\end{proposition}

\begin{remark} {\rm  By abuse of notation we  denote  the unitary representation  $\UU_{\rm f} \circ \pi$ on $SU(2)$
also by $\UU_{\rm f}$. }
\end{remark}

\begin{lemma}  \label{lem:rotinfock} Let $R \in SO(3)$   and $\kappa(R \cdot ) = \kappa(\cdot )$.
Then
\begin{align*}
\text{ (a)} \qquad &   \mathcal{U}_{\rm f}(R)  A(x)  \mathcal{U}^*_{\rm f}(R)  =  R^{-1}  A(R  x) ,  \\
\text{ (b)}\qquad &    \mathcal{U}_{\rm f}(R)  B(x)  \mathcal{U}^*_{\rm f}(R)   = R^{-1}  B(R  x) ,       \\
\text{ (c)}\qquad &   \mathcal{U}_{\rm f}(R)  E^\perp(x)  \mathcal{U}^*_{\rm f}(R)   = R^{-1}  E^\perp(R  x) .
\end{align*}
\end{lemma}
\begin{proof}
We observe that for $R \in SO(3)$ we find
\begin{align}
(\mathcal{U}_\mathfrak{v}(R) v_{x,b})(y) & = R v_{x,b}(R^{-1} y)  = \sum_{b'=1}^3R_{b',b} v_{Rx,b'}( y)   , \label{eq:trafoforho1}
\end{align}
where we used  Lemma~\ref{lem:reltinvpol} (c).
Using Eqs.~\eqref{eq:GammaasharpGamma} and \eqref{eq:trafoforho1} as well as Lemma~\ref{lem:reltinvpol-1}   we obtain
\begin{align*}
\mathcal{U}_{\rm f}(R) a^\#(g_{x,b}^{(\varepsilon)}) \mathcal{U}_{\rm f}^*(R) & = a^\#( \mathcal{U}_{\gggg}(R)   g_{x,b}^{(\varepsilon)})
=  a^\#( \tau_\varepsilon^{-1} \mathcal{U}_{\mathfrak{v}}(R)  \tau_\varepsilon  g_{x,b}^{(\varepsilon)})    \\
& =  a^\#( \tau_\varepsilon^{-1} \mathcal{U}_{\mathfrak{v}}(R)   v_{x,b})  =   \sum_{b'=1}^3R_{b',b}  a^\#( \tau_\varepsilon^{-1} v_{R x,b'})  \\
& =    \sum_{b'=1}^3R_{b',b}  a^\#( g_{Rx,b'}^{(\varepsilon)})
\end{align*}
This implies
$$
\mathcal{U}_{\rm f}(R) A_b(x) \mathcal{U}_{\rm f}^*(R) =  \sum_{b'=1}^3R_{b',b}  A_{b'}(R x)  .
$$
Thus (a) follows. Now (b) follows from (a) and by calculating the rotation. (c) Follows similarly as  (a) observing that $\omega$ is invariant under rotations.
\end{proof}

\begin{proposition}  \label{trafooffieldrot}  Let $U \in SU(2)$ and $R = \pi(U)$.
Then the following holds
\begin{align*}
\text{ (a)} \qquad & \mathcal{U}(U)  \hat{x}_j  \mathcal{U}(U)^*  \  =  R^{-1}  \hat{x}_j  , \\
\text{ (b)} \qquad &    \mathcal{U}(U)  p_j  \mathcal{U}(U)^*  \  =  R^{-1}  p_j  ,   \\
\text{ (c)}\qquad  &    \mathcal{U}(U)  \widehat{S}_{j} \mathcal{U}(U)^*    =  R^{-1}  \widehat{S}_{j} ,  \\ 
\text{ (d)} \qquad &    \mathcal{U}(U)  A(\hat{x}_j)  \mathcal{U}(U)^*  = R^{-1}   A(\hat{x}_j)  , \quad \text{ if } \kappa(R \cdot ) = \kappa(\cdot ) , \\
\text{ (e)}\qquad &    \mathcal{U}(U)  B(\hat{x}_j)  \mathcal{U}(U)^* = R^{-1}   B(\hat{x}_j)   , \quad \text{ if } \kappa(R \cdot ) = \kappa(\cdot )   ,  \\
\text{ (f)}\qquad &     \mathcal{U}(U)  E^\perp(\hat{x}_j)  \mathcal{U}(U)^* = R^{-1}  E^\perp(\hat{x}_j)  , \quad \text{ if } \kappa(R \cdot ) = \kappa(\cdot ) , \\
\text{ (g)} \qquad&  \mathcal{U}(U) H_{\rm f}   \mathcal{U}(U)^*  = H_{\rm f} ,  \\
\text{ (h)}\qquad &   \mathcal{U}(U) P_{\rm f}   \mathcal{U}(U)^*  = R^{-1}  P_{\rm f}     .
\end{align*}

\end{proposition}
\begin{proof}
Parts (a), (b), and (c)  are straight forward to verify. Parts (d)-(f) follow from (a) and Lemma~\ref{lem:rotinfock}.
Next we show (g)  and (h). Using Lemma \ref{lem:reltinvpol00} and the identity \eqref{eq:GammdGammaGamma} we find for any measurable $f : \R^3 \to \R$ and  $U \in SU(2)$ with $R = \pi(U)$
 \begin{align*} \mathcal{U}(U) d \Gamma(M_f)  \mathcal{U}(U)^*
&  = d \Gamma( \tau_\varepsilon^{-1}  \UU_\mathfrak{v}(\pi(U)) \tau_\varepsilon  M_f  \tau_\varepsilon^{-1}  \UU_\mathfrak{v}^*(\pi(U)) \tau_\varepsilon ) \\
& =  d \Gamma( \tau_\varepsilon^{-1}  \UU_\mathfrak{v}(R) F^{-1}   M_{f  }  F   \UU_\mathfrak{v}^*(R) \tau_\varepsilon ) \\
& =  d \Gamma( \tau_\varepsilon^{-1}   F^{-1}   M_{f \circ R^{-1} }  F   \tau_\varepsilon ) \\
& =  d \Gamma(   M_{f \circ R^{-1} }  )
\end{align*}
 Now choosing $f = \omega$ or $f : k \mapsto   k_j$ Parts (g) and (h) follow.
\end{proof}

In the following proposition we give a formula for the action of the rotation transformation in $\gggg$.

\begin{proposition}\label{prop:ActionOfRotation} For $R \in SO(3)$ define
\begin{align*}
\mathcal{D}^\mathcal{U}_{\lambda,\lambda'}(R;k) & := ( R^{-1} \varepsilon(k,\lambda) ) \cdot \varepsilon(R^{-1} k , \lambda')  .
\end{align*}
Then for $R \in SO(3)$
\begin{align}
\mathcal{D}^\mathcal{U}_{\lambda,\lambda'}(R^{-1};k) & = \mathcal{D}^\mathcal{U}_{\lambda',\lambda}(R; R k) \label{eq:relamongD}
\end{align}
and the following holds.
\begin{enumerate}[label=(\alph*)]
\item   \label{rotparta}  For any $h \in \gggg$
\begin{align} \label{eq:defofrotrep}
(\UU_{\gggg}(R) h)(k,\lambda) & =  \sum_{\lambda'=1,2}
\mathcal{D}^\mathcal{U}_{\lambda,\lambda'}(R;k)
h(R^{-1} k, \lambda').
\end{align}
\item \label{rotpartb}  In the sense of operator valued distributions for all $(k,\lambda) \in \R^3 \times \Z_2$
\begin{align*}
\mathcal{U}_{\rm f}(R)   a^\#(k,\lambda)    \mathcal{U}_{\rm f}(R)^*  =  \sum_{\lambda'=1,2}
\mathcal{D}^\mathcal{U}_{\lambda,\lambda'}(R^{-1} ;  k)  a^\#(R k,\lambda')
\end{align*}
\end{enumerate}
\end{proposition}
\begin{proof} 
Equation  \eqref{eq:relamongD}  follows from a straight forward calculation
using that the elements of $SO(3)$ preserve the inner product.  Now we  prove (a).
Using  the  property \eqref{eq:derink0} of
the Fourier transform, we find
 \begin{align*}
 (\UU_{\gggg}(R) h)({k},\lambda) & = \varepsilon({k},\lambda) \cdot  {F}  \left( {F}^{-1}
 \sum_{\lambda'=1,2}  R\varepsilon(\cdot, \lambda') h(\cdot,\lambda') \right)(R^{-1}{k}) \\
 & = \sum_{\lambda'=1,2} \varepsilon({k},\lambda) \cdot R \varepsilon( R^{-1} {k},\lambda')  h( R^{-1} {k},\lambda') .
 \end{align*}
(b)
We have by linearity and (a)
\begin{align*}
&  \sum_{\lambda=1,2} \int_{\R^3}  h(k,\lambda) \UU_{\rm f}(R)  a^*(k,\lambda)   \UU_{\rm f}^*(R)  dk    = \UU_{\rm f}(R) a^*(h) \UU_{\rm f}^*(R)=  a^*(\UU_{\gggg}(R) h) \\
& =\sum_{\lambda=1,2}  \int_{\R^3}   (\UU_{\gggg}(R) h)(k,\lambda)  a^*(k,\lambda) dk  \\
& =
 \sum_{\lambda, \lambda'=1,2} \int_{\R^3}
\mathcal{D}^\mathcal{U}_{\lambda,\lambda'}(R;k)
h(R^{-1} k, \lambda') a^*(k,\lambda)   dk    \\
& =
 \sum_{\lambda, \lambda'=1,2} \int_{\R^3}
\mathcal{D}^\mathcal{U}_{\lambda',\lambda}(R; R k)
h(k, \lambda) a^*(Rk, \lambda')  dk .
\end{align*}

Since $h \in \gggg$ is arbitrary the claim follows for $a^*(k,\lambda)$ in view of   \eqref{eq:relamongD}. Taking adjoints the claim then follows also for $a(k,\lambda)$.
\end{proof}

\subsection{Parity Symmetry}

Parity is the operation $x \mapsto - x$.
On the particle space we define
$$
\mathcal{P}_{{\rm p},s} :  L^2(\R^3; \mathcal{D}_s) \to L^2(\R^3;   \mathcal{D}_s ) , \quad  \psi \mapsto ( x \mapsto  \psi(-x) )
$$
for $s=0,1/2$.
 On the photon space we define
$$
\mathcal{P}_\mathfrak{v} :  \mathfrak{v} \to \mathfrak{v} , \quad v \mapsto ( x \mapsto  - v(-x) ) ,
$$
and  $$\mathcal{P}_{\gggg} = \tau_\varepsilon^{-1} \mathcal{P}_\mathfrak{v} \tau_\varepsilon  .
$$
We define
\begin{align*}
\PP_{\rm mat} & = \bigotimes_{j=1}^N \PP_{{\rm p},s_j} , \\
\PP_{\rm f} & =  \Gamma(\PP_{\mathfrak{g}}) ,  \\
\PP&= \PP_{\rm mat}  \otimes \PP_{\rm f}  .
\end{align*}

\begin{proposition}\label{prop:comm0} The maps $\PP_{\rm mat}$,  $\PP_{\rm f}$ and $\PP$ are unitary and commute with
the representations  $\mathcal{U}_{\rm mat   }$, $\mathcal{U}_{\rm f}$ and  $\mathcal{U}$, respectively.
\end{proposition}
\begin{proof} The unitarity property is straight forward to verify. The commutativity follows from the commutativity of $\PP_{\rm p}$ with $\UU_{\rm p}$
and  $\PP_{\vv}$ with $\UU_{\vv}$, which are straight forward to verify.
\end{proof}

 \begin{lemma} \label{lemmapartitysumfield}  Suppose $\kappa(- \cdot ) = \kappa(\cdot )$. Then
\begin{align*}
\text{ (a)} \qquad & \PP_{\rm f}  A(x) \PP_{\rm f}^*  = -  A(-x) , \\
\text{ (b)}\qquad &  \PP_{\rm f} B(x) \PP_{\rm f}^*  = B(-x) , \\
\text{ (c)}\qquad &  \PP_{\rm f}  E^\perp(x) \PP_{\rm f}^*  = - E(-x)  .
\end{align*}
\end{lemma}

\begin{proof}
We observe that for $S=- \one_{3\times 3}$ we find from \eqref{eq:trafoforho05}
\begin{align}
(\mathcal{P}_\mathfrak{v} v_{x,b})(y) & = - v_{x,b}(- y)  = - v_{-x,b}( y)   . \label{eq:trafoforho}
\end{align}
Now we find similar as in the proof of Lemma~\ref{lem:rotinfock} using Lemma \ref{lem:reltinvpol-1} and \eqref{eq:trafoforho}
\begin{align*}
\mathcal{P}_{\rm f } a^\#(g_{x,b}^{(\varepsilon)}) )\mathcal{P}_{\rm f}^*  & = a^\#( \mathcal{P}_{\mathfrak{g}}   g_{x,b}^{(\varepsilon)})
=  a^\#( \tau_\varepsilon^{-1} \mathcal{P}_{\mathfrak{v}}  \tau_\varepsilon  g_{x,b}^{(\varepsilon)})    \\
& =  a^\#( \tau_\varepsilon^{-1} \mathcal{P}_{\mathfrak{v}}   v_{x,b})  =
  a^\#( -  \tau_\varepsilon^{-1} v_{- x,b})  \\
& =   -  a^\#( g_{-x,b}^{(\varepsilon)})  .
\end{align*}
This implies
$$
\mathcal{P}_{\rm f} A_b(x) \mathcal{P}_{\rm f}^* =  -   A_{b}(-  x)  .
$$
Thus (a) follows. Now (b) follows from (a) and by calculating the rotation. (c) Follows similarly as in (a) observing that $\omega(-\cdot ) = \omega$.
\end{proof}

In view of the following proposition we see that  $\PP$ has the physical interpretation of
parity inversion.

 \begin{proposition} \label{trafooffieldpar}
$\PP$ has satisfies the following properties.
\begin{align*}
\text{ (a)} \qquad & \mathcal{P}  \hat{x}_j   \PP  = -\hat{x}_j , \\
\text{ (b)} \qquad &    \PP p_j  \PP^*  = -  p_j ,  \\
\text{ (c)}\qquad  &   \PP \widehat{S}_{j} \PP^*  =   \widehat{S}_{j} ,\\
\text{ (d)} \qquad & \PP A(\hat{x}_j) \PP^*  = -  A(\hat{x}_j)     , \quad \text{ if } \kappa(- \cdot ) = \kappa(\cdot ) ,  \\
\text{ (e)}\qquad &  \PP B(\hat{x}_j) \PP^*  = B(\hat{x}_j)  , \quad \text{ if } \kappa(- \cdot ) = \kappa(\cdot ) , \\
\text{ (f)}\qquad &  \PP E^\perp(\hat{x}_j) \PP^*  = - E^\perp(\hat{x}_j)  , \quad \text{ if } \kappa(- \cdot ) = \kappa(\cdot )  , \\
\text{ (g)} \qquad&  \PP H_{\rm f} \PP^*  = H_{\rm f} , \\
\text{ (h)}\qquad &  \PP P_{\rm f}  \PP^*  = - P_{\rm f} .
\end{align*}

\end{proposition}

\begin{proof} The proof is analogous to that of Proposition \ref{trafooffieldrot}.
\end{proof}

In the following proposition we give a formula for the action of the parity  in $\gggg$.
\begin{proposition} The map  $\mathcal{P}_{\gggg}$ has the following properties.
Define
\begin{equation*}
\mathcal{D}^\mathcal{P}_{\lambda,\lambda'}({k}) := -  \varepsilon({k},\lambda) \cdot   \varepsilon(-{k},\lambda') .
\end{equation*}
Then $\mathcal{D}^\mathcal{P}_{\lambda,\lambda'}({k})=\mathcal{D}^\mathcal{P}_{\lambda',\lambda}(-{k})$.
\begin{itemize}
\item[(a)]   For any $h \in \gggg$ we have
 for almost all $({k},\lambda ) \in \R^3 \times \{1,2\}$
\begin{equation*} 
(\PP_{\gggg} h)({k},\lambda) = \sum_{\lambda'=1,2} \mathcal{D}_{\lambda, \lambda'}^\mathcal{P}({k})
h(-{k},\lambda') .
\end{equation*}

\item[(b)] We  have in the sense of operator valued distributions for all $(k,\lambda) \in \R^3 \times \Z_2$

$$ \PP_{\rm f} a^\#(k,\lambda) \PP_{\rm f }^* = \sum_{\lambda'=1,2} \mathcal{D}^\mathcal{P}_{\lambda,\lambda'}({k})  a^\#(-k,\lambda')$$
\end{itemize}
\end{proposition}
\begin{proof} 
The first statement follows from the symmetry of the scalar product.
 (a)
Using    \eqref{eq:derink0}, we find
 \begin{align*}
 (\PP_{\gggg} h)({k},\lambda) & = \varepsilon({k},\lambda) \cdot  {F}  \left( {F}^{-1}
 \sum_{\lambda'=1,2} ( - \varepsilon(\cdot, \lambda'))  h(\cdot,\lambda') \right)(-{k}) \\
 & =  \sum_{\lambda'=1,2}  (- \varepsilon({k},\lambda) ) \cdot \varepsilon(-{k},\lambda')  h(-{k},\lambda') .
 \end{align*}
(b) We have by linearity and  (a)
\begin{align*}
&  \sum_{\lambda=1,2} \int h(k,\lambda) \PP_{\rm f}  a^*(k,\lambda)   \PP_{\rm f}^*  dk    = \PP_{\rm f} a^*(h) \PP_{\rm f}^* =  a^*(\PP_{\gggg} h) \\
& =\sum_{\lambda=1,2} \int_{\R^3}   (\PP_{\gggg} h)(\lambda,k)  a^*(k,\lambda)  dk \\
& =
 \sum_{\lambda, \lambda'=1,2} \int_{\R^3}    \mathcal{D}_{\lambda, \lambda'}^\mathcal{P}({k})a^*(k,\lambda)
h(-{k},\lambda') dk \\
& = \sum_{\lambda, \lambda'=1,2} \int_{\R^3}  \mathcal{D}_{\lambda', \lambda}^\mathcal{P}(-{k})a^*(-k,\lambda')h({k},\lambda') dk
\end{align*}
Since $h \in \gggg$ is arbitrary the claim follows for $a^*(k,\lambda)$. Taking adjoints the claim then follows also for $a(k,\lambda)$.
\end{proof}

\subsection{Time reversal symmetry}

\label{subsecteim}

We define  time reversal symmetry.
Let $K$ denote complex conjugation on $L^2(\R^3; \mathcal{D}_s)$.
Define the   operators
\begin{align*}
&
   \mathcal{T}_{{\rm p},s} := \left\{ \begin{array}{ll}  K     & , \quad \text{ if } \quad  s = 0  , \\
  (K   \sigma_2 )   &  , \quad \text{ if }  \quad s = 1/2   \end{array} \right.
 \end{align*}
and
\begin{align*}
&
   \mathcal{T}_{\rm mat} :=  \bigotimes_{j=1}^N    \mathcal{T}_{{\rm p},s_j} .
 \end{align*}
Let $\mathcal{K}_{\mathfrak{v}}$ denote complex conjugation in $\mathfrak{v}$,  and  let 
\begin{equation} \label{eq:defofkh}
\mathcal{K}_{\gggg} =   \tau_\varepsilon^{-1}  \mathcal{K}_{\mathfrak{v}} \tau_\varepsilon
\end{equation}
denote its action on  $\gggg$.
 Next we define operator of time reversal on the quantum field
\begin{equation} \label{eq:defoftf}
\mathcal{T}_{\rm f} := \Gamma(- \mathcal{K}_\mathfrak{g}) . 
\end{equation}
We define the operator of time reversal in  the full Hilbert space by
\begin{equation} \label{eq:defoftimerev}
\mathcal{T} = \mathcal{T}_{\rm mat} \otimes \mathcal{T}_{\rm f} .
\end{equation}

\begin{proposition}  \label{lem:propoftat}\label{prop:comm}
 The maps  $\mathcal{T}_{\rm mat}$, $\TT_{\rm f}$, and $\TT$ are anti-unitary operators, which commute with the representations $\UU_{\rm mat}$, $\UU_{\rm f}$, and $\UU$
and the operators $\PP_{\rm mat}$, $\PP_{\rm f}$, and $\PP$, respectively. We have  $\TT_{\rm f}^2 = 1$, and
\begin{align*}
& \TT_{\rm mat}^2 = (-1)^{\sum_{j=1}^N 2 s_j}  , \qquad \qquad\, \TT^2 =  (-1)^{\sum_{j=1}^N 2 s_j}  .
\end{align*}
\end{proposition}
\begin{proof} The anti-unitarity is straight forward to verify on the one particle spaces. On the tensor product it then follows by Lemma   \ref{lem:antiline}.
The commutativity can be seen by verifying it   on the one particle spaces.  The last statement follows
from 
$$
\mathcal{T}_{\rm mat}^2 = \bigotimes_{j=1}^N  (  \mathcal{T}_{{\rm mat},s_j})^2 
$$
with  $(  \mathcal{T}_{{\rm mat},0})^2 = 1$ and  $(  \mathcal{T}_{{\rm mat},1/2})^2  = ( K \sigma_2 ) ( K \sigma_2) = K^2 \sigma_2 (-\sigma_2) =  -1$,
\end{proof}

 \begin{lemma} \label{lem:lemmatime}  Suppose $ \overline{ \kappa( \cdot )}  = \kappa(- \cdot ) $. Then the following holds
\begin{align*}
\text{ (d)} \qquad & \TT_{\rm f} A(x) \TT_{\rm f}^*  = -  A(x) ,  \\
\text{ (e)}\qquad &  \TT_{\rm f} B(x) \TT_{\rm f}^*  = - B(x) ,  \\
\text{ (e)}\qquad &  \TT_{\rm f} E^\perp(x) \TT_{\rm f}^*  =  E^\perp(x)  .
\end{align*}
\end{lemma}
\begin{proof}
It follows directly from the definition,  a trivial change of variables, and the assumption about $\kappa$  that
\begin{align}
(\mathcal{K}_\mathfrak{v} v_{x,b})(y) & =  v_{x,b}( y)\,.
\end{align}
Now we find using  Lemma~\ref{lem:reltinvpol-1}
\begin{align*}
 \Gamma(- \mathcal{K}_\gggg)   a^*(g_{x,b}^{(\varepsilon)})  \Gamma(- \mathcal{K}_\gggg)^*  & = a^*( - \mathcal{K}_{\gggg}   g_{x,b}^{(\varepsilon)})
=  -a^*( \tau_\varepsilon^{-1} \mathcal{K}_{\mathfrak{v}}  \tau_\varepsilon  g_{x,b}^{(\varepsilon)})    \\
& = - a^*( \tau_\varepsilon^{-1} \mathcal{K}_{\mathfrak{v}}   v_{x,b})  =
  - a^*(   \tau_\varepsilon^{-1} v_{x,b})  \\
& =   -  a^*( g_{x,b}^{(\varepsilon)})
\end{align*}
This implies  $\mathcal{T}_{\rm f} a^*(g_{x,b}^{(\varepsilon)})   \mathcal{T}_{\rm f}^* = - a^*( g_{x,b}^{(\varepsilon)})  $ and by taking adjoints
 $\mathcal{T}_{\rm f} a(g_{x,b}^{(\varepsilon)})   \mathcal{T}_{\rm f}^* = - a( g_{x,b}^{(\varepsilon)})  $.
 Hence
$$
\mathcal{T}_{\rm f}A_b(x) \mathcal{T}_{\rm f}^* =  -   A_{b}(  x)  .
$$
This shows  (a). Now (b) follows from (a) and by calculating the rotation. (c) Follows similarly as in (a) observing that $i \omega$ changes sign when  complex
conjugating.
\end{proof}

In view of the following proposition we see that  $\TT$ has the physical interpretation of
time reversal.

 \begin{proposition} \label{proptime} \label{lem:statet} Suppose   $ \overline{ \kappa( \cdot )}  = \kappa(- \cdot ) $. Then $\TT$ is anti-unitary and satisfies the following properties
\begin{align*}
\text{ (a)} \qquad & \mathcal{T}  \hat{x}_j   \TT^*  = \hat{x}_j , \\
\text{ (b)} \qquad &    \TT p_j  \TT^*  = -  p_j ,  \\
\text{ (c)}\qquad  &   \TT \widehat{S}_{j } \TT^*  = -  \widehat{S}_{j } ,\\
\text{ (d)} \qquad & \TT A(\hat{x}_j) \TT^*  = -  A(\hat{x}_j)   , \quad \text{ if } \overline{\kappa(- \cdot )} = \kappa(\cdot )  ,   \\
\text{ (e)}\qquad &  \TT B(\hat{x}_j) \TT^*  = - B(\hat{x}_j)  , \quad \text{ if } \overline{ \kappa(- \cdot )} = \kappa(\cdot )  ,\\
\text{ (f)}\qquad &  \TT E^\perp(\hat{x}_j) \TT^*  =  E^\perp(\hat{x}_j)  , \quad \text{ if } \overline{\kappa(- \cdot )} = \kappa(\cdot )  , \\
\text{ (g)} \qquad&  \TT H_{\rm f} \TT^*  = H_{\rm f} , \\
\text{ (h)}\qquad &  \TT P_{\rm f}  \TT^*  = - P_{\rm f}  .
\end{align*}
\end{proposition}

\begin{proof} Parts
(a), (b),  and (c)  are straight forward to verify. Parts  (d), (e), and  (f) follow from Lemma  \ref{lem:lemmatime}.
 Using Lemma \ref{lem:reltinvpol00} we find for any measurable $f : \R^3 \to \R$
 \begin{align*} \TT_{\rm f} d \Gamma(M_f)  \TT_{\rm f}^*
&  = d \Gamma( \tau_\varepsilon^{-1}  \KK_\mathfrak{v} \tau_\varepsilon   M_f  \tau_\varepsilon^{-1}  \KK_\mathfrak{v}^* \tau_\varepsilon ) \\
& =  d \Gamma( \tau_\varepsilon^{-1}  \KK_\mathfrak{v} F^{-1}   M_{f  }  F   \KK_\mathfrak{v}^* \tau_\varepsilon ) \\
& =  d \Gamma( \tau_\varepsilon^{-1}   F^{-1}   M_{f(-\cdot)   }  F   \tau_\varepsilon ) \\
& =  d \Gamma(   M_{f(- \cdot)   }  ) ,
\end{align*}
where in the third equality we used that  the Fourier transform satisfies  the following properties
$F \overline{\varphi} = \overline{ F \varphi }(- \cdot)$ and $ \overline{ F^{-1} \varphi } = F^{-1} \overline{\varphi}(-\cdot)$ for $\varphi \in L^2(\R^3)$.
 Now choosing $f = \omega$ or $f : k \mapsto   k_j$ Parts (g) and (h) follow.
\end{proof}

In the following proposition we give a formula for the action of the time reversal symmetry  in $\gggg$.
\begin{proposition} 
For $h \in \gggg$ we have  for almost all $({k},\lambda ) \in \R^3 \times \{1,2\}$
\begin{equation*} 
(\mathcal{K}_\gggg h)({k},\lambda) = \sum_{\lambda'=1,2}  \mathcal{D}^\mathcal{T}_{\lambda, \lambda'}({k})
\overline{h(-{k},\lambda')} ,
\end{equation*}
where
$\mathcal{D}^\mathcal{T}_{\lambda, \lambda'}({k}) := \varepsilon({k},\lambda) \cdot \varepsilon(-{k},\lambda') $.
Then $\mathcal{D}^\mathcal{T}_{\lambda, \lambda'}({k}) = \mathcal{D}^\mathcal{T}_{\lambda', \lambda}(-{k})$ and
in the sense of operator valued distributions for all $(k,\lambda) \in \R^3 \times \Z_2$
$$
\mathcal{T}_{\rm f} a^\#(k,\lambda) \mathcal{T}_{\rm f}^* = - \sum_{\lambda'=1,2} \mathcal{D}^{\mathcal{T}}_{\lambda,\lambda'}(k)
a^\#(-k,\lambda') .
$$
\end{proposition}
\begin{proof}
Using for $\varphi \in L^2(\R^3)$ the following property of
the Fourier transform
 ${F} \overline{({F}^{-1} \varphi)}({k}) = \overline{\varphi}(-{k})$,
 we find
 \begin{align*}
 (\mathcal{K}_\gggg  h)({k},\lambda) & = \epsilon({k},\lambda) \cdot  {F} \overline{ \left( {F}^{-1}
 \sum_{\lambda'=1,2} \varepsilon(\cdot, \lambda') h(\cdot,\lambda') \right)}({k}) \\
 & = \sum_{\lambda'=1,2} \varepsilon({k},\lambda) \cdot \varepsilon(-{k},\lambda') \overline{ h(-{k},\lambda')} .
 \end{align*}
This shows the first identity.
Using this,   we find  by anti-linearity
\begin{align*}
&  \sum_{\lambda=1,2} \int \overline{h(k,\lambda)} \TT_{\rm f}  a^*(k,\lambda)   \TT_{\rm f}^*  dk    = \TT_{\rm f} a^*(h) \TT_{\rm f}^* =  a^*(- \KK_{\gggg} h) \\
& = - \sum_{\lambda=1,2} \int_{\R^3}   (\KK_{\gggg} h)(\lambda,k)  a^*(k,\lambda)  dk \\
& =
- \sum_{\lambda, \lambda'=1,2} \int_{\R^3}  \overline{h(-{k},\lambda')}  \mathcal{D}_{\lambda, \lambda'}^\mathcal{T}({k})  a^*(k,\lambda)
 dk \\
& = -  \sum_{\lambda, \lambda'=1,2} \int_{\R^3} \overline{h({k},\lambda)} \mathcal{D}_{\lambda', \lambda}^\mathcal{T}(-{k})a^*(-k,\lambda')  dk .
\end{align*}
Since $h \in \gggg$ is arbitrary the second identity  follows for $a^*(k,\lambda)$. Taking adjoints the claim then follows also for $a(k,\lambda)$.
\end{proof}

\section{Hamiltonians with Symmetries} \label{sec:ham}

In this section we consider Hamiltonians of non-relativistic qed, and discuss their symmetry properties.

\begin{theorem}\label{rotfullop}   Suppose $U \in SU(2)$,  $ R = \pi(U)$,  $  V(x_1,...,x_N)  = V(Rx_1,...,Rx_N)$ for all $x_1,...,x_N \in \R^3$,    $B_{\rm ext}(x) = R B_{\rm ext}( R^{-1} x) $ for all $x \in \R^3$,
and $\kappa(R \cdot ) = \kappa(\cdot )$.
Then
\begin{align*}
   \mathcal{U}(U) H  \mathcal{U}(U)^* = H  .
\end{align*}
\end{theorem}

\begin{proof}
Using \eqref{exprforAext},  properties of the cross product,  a change of variables, and the symmetry properties of $B_{\rm ext}$  we find
\begin{align*}
R A_{{\rm ext}}(R^{-1} x)  & =  -  \int \frac{(x- R y) \times  R B_{\rm ext}(y)}{4\pi |x- R y|^3} dy = A_{\rm ext}(x) .
\end{align*}
Thus using Proposition~\ref{trafooffieldrot}
\begin{align*}
& \UU(U) H  \UU(U)^* \\
& = \sum_{j=1}^N  \left\{ \frac{1}{2m_j} \left(   R^{-1} p_j   +  q_j (  R^{-1} A(\hat{x}_j)   +A_{\rm ext}( R^{-1}\hat{x}_j) ) \right)^2   + \mu_j   R ^{-1} \widehat{S}_j \cdot  ( R^{-1}   B(\hat{x}_j) + B_{\rm ext}(R^{-1} \hat{x}_j)) \right\} \\
& \quad  +  H_{\rm f}   + V(R^{-1}  \hat{x}_1,...,R^{-1} \hat{x}_N)  \\
& = H  ,
\end{align*}
where in the last line we used the assumed properties of $B_{\rm ext}$ and $V$.
\end{proof}

\begin{theorem}
Suppose  $  V(x_1,...,x_N)  = V(-x_1,...,-x_N)$ for all $x_1,...,x_N \in \R^3$,  $B_{\rm ext}(\cdot ) =  B_{\rm ext}( - \cdot ) $, and $\kappa(-\cdot) = \kappa(\cdot)$.
Then
\begin{align*}
 \PP H \PP^*  = H .
\end{align*}
\end{theorem}
\begin{proof}
Using  \eqref{exprforAext},  the properties of the cross product,  a change of variables, and the symmetry properties of $B_{\rm ext}$  we find 
\begin{align*}
A_{{\rm ext}}(- x)  & =  -  \int \frac{(-(x -  y)) \times   B_{\rm ext}(-y)}{4\pi |x-  y|^3} dy = - A_{\rm ext}(x) .
\end{align*}
Thus we find  from Proposition \ref{trafooffieldpar}
\begin{align*}
& \mathcal{P}  H_{} \mathcal{P}^* \\
& = \sum_{j=1}^N  \left\{ \frac{1}{2m_j} \left(   -  p_j   -   q_j    A(\hat{x}_j)   + q_j  A_{\rm ext}(  - \hat{x}_j) ) \right)^2   + \mu_j   \widehat{S}_j \cdot (   B(\hat{x}_j) + B_{\rm ext}(- \hat{x}_j)) \right\} \\
& +  H_{\rm f}   + V(- \hat{x}_1,...,- \hat{x}_N)  \\
& = H_{} ,
\end{align*}
where in the last line we used the assumed properties of $B_{\rm ext}$ and $V$.
\end{proof}
\begin{theorem}   \label{eq:thmbzerokr}
Suppose    $B_{\rm ext} =0 $ and $\overline{\kappa(\cdot)} = \kappa(-\cdot)$. Then
\begin{align*}
 \TT H \TT^*  = H .
\end{align*}

\end{theorem}
\begin{proof}
 We find  from Proposition  \ref{proptime}
\begin{align*}
& \mathcal{T}  H \mathcal{T}^*
 = \sum_{j=1}^N  \left\{\frac{1}{2m_j} \left(   -  p_j   -   q_j    A(\hat{x}_j) \right)^2   + \mu_j   \widehat{S}_j \cdot   B(\hat{x}_j)  \right\} +  H_{\rm f}   + V(\hat{x}_1,...,\hat{x}_N)   = H .
\end{align*}
\end{proof}

\begin{theorem}   \label{eq:thmbzerokr2}
If $  V(x_1,...,x_N)  = V(-x_1,...,-x_N)$ for all $x_1,...,x_N \in \R^3$,   $B_{\rm ext}(\cdot ) = - B_{\rm ext}( - \cdot) $,    and $\overline{\kappa(\cdot)} = \kappa(-\cdot) = \kappa(\cdot)$. Then
\begin{align*}
\TT  \PP  H (\TT \PP)^*  = H .
\end{align*}
\end{theorem}
\begin{proof}
Using \eqref{exprforAext}, the    properties of the cross product,  a change of variables, and the symmetry properties of $B_{\rm ext}$  we find
\begin{align*}
A_{{\rm ext}}(- x)  & =  -  \int \frac{(-(x -  y)) \times   B_{\rm ext}(-y)}{4\pi |x-  y|^3} dy =  A_{\rm ext}(x) .
\end{align*}
 Thus we find  from Propositions \ref{trafooffieldpar} and   \ref{proptime}
\begin{align*}
& \TT  \mathcal{P}  H \mathcal{P}^* \ \TT^* \\
& = \TT \bigg(  \sum_{j=1}^N  \left\{ \frac{1}{2 m_j}  \left(   -  p_j   -   q_j    A(\hat{x}_j)   + q_j  A_{\rm ext}(  - \hat{x}_j)  \right)^2   + \mu_j   \widehat{S}_j  \cdot  (   B(\hat{x}_j) + B_{\rm ext}(- \hat{x}_j)) \right\} \\
& +  H_{\rm f}   + V(- \hat{x}_1,...,- \hat{x}_N)  \bigg) \TT^* \\
& =  \sum_{j=1}^N  \left\{\frac{1}{2 m_j}  \left(    p_j   +    q_j    A(\hat{x}_j)   + q_j  A_{\rm ext}( \hat{x}_j) ) \right)^2   + \mu_j   \widehat{S}_j (   B(\hat{x}_j) - B_{\rm ext}(-\hat{x}_j)) \right\} \\
& +  H_{\rm f}   + V(- \hat{x}_1,...,- \hat{x}_N)    \\
& = H ,
\end{align*}
where  we  used the assumed  properties of $B_{\rm ext}$ and $V$.
\end{proof}

As an application of the abstract
Kramer theorem,  we now show  the following  degeneracy result.

\begin{theorem}  \label{thm:kramerexp}
Suppose $\sum_{j=1}^N 2 s_j $ is odd, and  that at least one  of the   following two assumptions hold.
\begin{itemize}
\item[(i)] $B_{\rm ext} = 0$  and $\overline{\kappa(\cdot)} = \kappa(-\cdot)$
\item[(ii)]  $V(-x_1,...,-x_N) = V(x_1,...,x_N)$ and $B_{\rm ext}(-x) = - B_{\rm ext}(x)$,  and $\overline{\kappa(\cdot)} = \kappa(-\cdot) = \kappa(\cdot)$.
\end{itemize}
Then,  any  eigenvalue   of   $H$  is at least two fold degenerate. If the multiplicity of an eigenvalue is finite, it is    even.
\end{theorem}
\begin{proof}
In case (i) the assertion follows from  Kramers degeneracy theorem~\ref{le:kram} for $\theta= \mathcal{T}$,   Proposition \ref{prop:comm}, and  Theorem~\ref{eq:thmbzerokr}.
In case (ii) the assertion follows  from  Kramers degeneracy theorem~\ref{le:kram} for $\theta= \mathcal{T} \mathcal{P}$,   Proposition \ref{prop:comm}, and  Theorem~\ref{eq:thmbzerokr2}.
\end{proof}

\begin{remark} {\rm \quad
\begin{enumerate}[(a)]
\item  We note that Theorem~\ref{thm:kramerexp} for the case $N=1$, $s_1=1/2$,   and  (i) with the additional assumption  $V(-x) = V(x)$
was shown in \cite{LosMiySpo09,LosMiySpo10}.
Thus  Theorem \ref{thm:kramerexp} relaxes  the unnecessary parity-symmetry assumption for the external potential $V$.
In fact, the proof given in \cite{LosMiySpo09} uses   the symmetry  $\mathcal{P} \mathcal{T}$,
while  the proof  in  \cite{LosMiySpo10} uses the symmetry    $\mathcal{T}$  in the so called Schr\"odinger representation,
cf.  Section  \ref{sec:schr} of this paper. 
\item Since  the  classical Kramer theorem  uses time inversion symmetry it cannot be applied to situations with external magnetic fields.
However if one considers  the anti-linear symmetry  $\mathcal{P} \mathcal{T}$ one can include external magnetic fields, which satisfy
a symmetry condition. We note that the result (ii)  also holds for an ordinary Schr\"odinger operator without any quantized
electromagnetic field, as the proof also applies to such a situation with a straight forward (trivial) modification of the proof.
\end{enumerate}
}
\end{remark}

Next we consider the restriction to symmetric subspaces. To this end we introduce notation satisfying the following hypothesis.

\begin{hypothesis} \label{sym} {\rm  The set  $\mathfrak{P} = \{p_1,....,p_L\}$, $L \in \N \cap \{1,...,N\}$, is a partition   of $\{1,....,N\}$ such that on each element $p \in \mathfrak{P}$
of the partition
the  numbers  $m_j$, $s_j$, $q_j$, and  $\mu_j$ are equal (cf.  \eqref{defofham}).  The function  $\tau$  maps $\mathfrak{P}$ to  $\{0,1\}$.
The potential $V$ is symmetric with respect to interchange of particle coordinates of particles which belong to the same element $p \in \mathfrak{P}$.  }
\end{hypothesis}

\begin{remark} {\rm The function $\tau$ in Hypothesis \ref{sym} is used to  specify the statistics of identical particles.
The value $0$ will  be used to describe bosons while the value 1 will be used  describe fermions. By physical laws,
 spin zero particles are bosons while
spin 1/2 particles are fermions. }
\end{remark}

For a  finite set $S$ we shall denote by $\mathfrak{S}_S$ the set of all permutations of the set $S$.
 For a subset $S \subset\{1,...,N\} $    and $\sigma \in  \mathfrak{S}_S$
we denote by    $\underline{\sigma}$ its extension to  $\{1,...,N\}$
by setting it equal to the identity  on $\{1,....,N\} \setminus S$.
Suppose  the partition $\mathfrak{P}$ satisfies Hypothesis \ref{sym}.  Then  for any $p \in \mathfrak{P}$ and $\sigma \in \mathfrak{S}_p$
it follows that $\mathfrak{U}(\underline{\sigma})$, defined in \eqref{defofsymmop},   leaves $\HH_{\rm mat}$ invariant,
and we can define the subspace
\begin{align} \label{ccc}
\HH_{{\rm mat},\mathfrak{P},\tau} = \{ \psi \in \HH_{\rm mat} :  \forall p \in \mathfrak{P} , \forall \sigma \in \mathfrak{S}_p , \mathfrak{U}(\underline{\sigma})  \psi = {\rm sgn}(\sigma)^{\tau(p)} \psi \} ,
\end{align}
where ${\rm sgn} (\sigma)$ defines the signum of the permutation  $\sigma$.
Furthermore, it follows from the definitions that $\mathfrak{U}(\underline{\sigma})$
commutes with the symmetries
$\mathcal{U}_{\rm mat}$,   $\mathcal{P}_{\rm mat}$,   $\mathcal{T}_{\rm mat}$ as well as
the Hamiltonian  $H$. In particular, $\HH_{{\rm mat},\mathfrak{P},\tau}\otimes \FF_s(\mathfrak{g}) $ is an invariant subspace of $H$.

\begin{theorem} \label{thm:kramerexp2}  Suppose that  the  partition $\mathfrak{P}$, the function $\tau$
 and the potential $V$,    satisfy Hypothesis  \ref{sym}.
Suppose $\sum_{j=1}^N 2 s_j $ is odd, and  (i) or (ii) of Theorem \ref{thm:kramerexp} holds.
Then,  any eigenvalue     of   $H |_{  \HH_{{\rm mat},\mathfrak{P},\tau} \otimes \FF_s(\mathfrak{g})}$ has even or infinite multiplicity.
\end{theorem}
\begin{proof} Follows from the same proof as  Theorem \ref{thm:kramerexp},  by observing in addition that $\mathcal{T}$ and $\mathcal{P}$ commute with $\mathfrak{U}(\underline{\sigma})$
for any $\sigma \in \mathfrak{S}_p$ and $p \in \mathfrak{P}$, and thus leave  $\HH_{{\rm mat},\mathfrak{P},\tau}$ invariant.
\end{proof}

\begin{remark} {\rm We note that Theorem~\ref{thm:kramerexp2} for the special case $\mathfrak{P} = \{ p \}$ with $p= \{1,....,N\}$, $s_j =1/2$ for all $j \in p$, and  $\tau(p) = 1$, and    with the additional assumption that $V$ is given by the Coulomb potential of $N$ electrons in the  presence
of the electric field of a  nucleus
was shown in \cite{LosMiySpo09}. }
\end{remark}

\section{Translationally invariant  Hamiltonians}  \label{sec:trainv}
We write  the Hamiltonian \eqref{defofham}
acting in the Hilbert space
$
\HH_{\rm mat}   \otimes \FF_s(\mathfrak{g})
$
in the following notation
$$
H =
\sum_{j=1}^N T_j
+ H_{\rm f}  + V(\widehat{x}_1,...,\widehat{x}_N)  , \quad   T_j := \frac{1}{2 m_j} (p_j + q_j A(\widehat{x}_j))^2 + \mu_j \widehat{S}_j \cdot B(\widehat{x}_j) ,
$$
and  we assume that there is no external magnetic field.
Furthermore, we assume that  the potential $V$   in the definition of the Hamiltonian \eqref{defofham}  is  translationally  invariant, i.e., that for all $a \in \R^3$
\begin{align} \label{translinv}
V(x_1 + a,..., x_N + a) =   V(x_1,..., x_N) .
\end{align}
Using the unitary transformation
$$
 U = \exp( i {x}_N \cdot  (P_{\rm f}  + \sum_{j=1}^{N-1} p_j ) )
$$
and a Fourier transform in the variable $x_N$ we can write
$$
H  = \int^\oplus_{\R^3} H(\xi) d \xi ,
$$
where
$$
H(\xi) := \frac{1}{2m_N}  (\xi  - \sum_{j=1}^{N-1} p_j  - P_{\rm f}  +  q_N A(0))^2 + \mu_N  \widehat{S}_N \cdot B( 0 ) +
\sum_{j=1}^{N-1} T_j
+ H_{\rm f}  + V(\widehat{x}_1,...,\widehat{x}_{N-1},0)
$$
acts in
\begin{align} \label{defofHprime}
 \HH_{\rm mat}'  \otimes  \mathcal{D}_{s_N} \otimes \FF_s(\mathfrak{g})   ,
\end{align}
where
$$\HH_{\rm mat}'  :=  \bigotimes _{j=1}^{N-1}
L^2(\R^3 ; \mathcal{D}_{s_j} ) ,
$$
cf. \cite{LosMiySpo07,HasHer08-1}. We define  $\mathcal{U}_{\rm  mat}'$, $\mathcal{P}_{\rm mat}'$, and $\mathcal{T}_{\rm mat}'$  on $ \HH_{\rm mat}'$
as in Section  \ref{sec:sym}. On \eqref{defofHprime} we define the symmetries
\begin{align*}
\mathcal{U}'(U)  & := \mathcal{U}_{\rm  mat}'(U) \otimes  D_{s_N}(U) \otimes  \mathcal{U}_{\rm f}(\pi(U)) , \quad U \in SU(2) \\
\mathcal{P}'  & := \mathcal{P}_{\rm  mat}' \otimes  \one_{\mathcal{D}_{s_N}}  \otimes \mathcal{P}_{\rm f} \\
\mathcal{T}'  & := \mathcal{T}_{\rm  mat}' \otimes  \mathcal{T}'_{{\rm p},s}   \otimes \mathcal{T}_{\rm f} ,
\end{align*}
where we defined
\begin{align*}
&
   \mathcal{T}'_{{\rm p},s} := \left\{ \begin{array}{ll}  K_s     & , \quad \text{ if } \quad  s = 0  , \\
  (K_s   \sigma_2 )   &  , \quad \text{ if }  \quad s = 1/2   \end{array} \right.
 \end{align*}
where $K_{s}$ denotes complex conjugation on $\mathcal{D}_s = \C^{2s+1}$.

\begin{lemma} \label{lem:invt} Suppose $V$ is translationally invariant, cf.  \eqref{translinv}.
\begin{enumerate}[(a)]
\item  Let  $ U \in SU(2) $, $R = \pi(U)$,   $V(Rx_1,...,Rx_N,0)  = V(x_1,...,x_N,0) $ for all $x_j \in \R^3$,  and
 $\kappa(\cdot)  = \kappa( R \cdot) $.   Then   for all $\xi \in \R^3$
$$ \mathcal{  U }'(U)  H(\xi)  \mathcal{  U }'(U)^{*} = H( R \xi )  .$$
\item Let  $V(x_1,...,x_{N-1},0)= V(-x_1,...,-x_{N-1},0) $  for all $x_j \in \R^3$   and  $\kappa(\cdot)  = \kappa(-\cdot)$.  Then for all $\xi \in \R^3$  $$ \mathcal{  P }'  H(\xi)  {\mathcal{  P }'}^{*} = H(-\xi ) . $$
\item  If $ \overline{\kappa(\cdot)}  = \kappa(-\cdot) $, then for all $\xi \in \R^3$
$$ \mathcal{  T }' H(\xi) {\mathcal{T }'}^{*} = H(-\xi)  .$$
\end{enumerate}
\end{lemma}
\begin{proof}
The Lemma follows as a consequence of Lemmas~\ref{lem:rotinfock},   \ref{lemmapartitysumfield}, and  \ref{lem:lemmatime}   and Propositions~\ref{trafooffieldrot}, \ref{trafooffieldpar}, and \ref{lem:statet},   respectively, and their trivial  adaption to \eqref{defofHprime}.
\end{proof}

\begin{theorem}\label{thm:transinv1} Suppose  $V$ is translationally invariant and   $\sum_{j=1}^N 2 s_j$ is odd. If $\overline{\kappa(\cdot)} = \kappa(-\cdot)$ each eigenvalue of $H(0)$
has even or infinite multiplicity. If in addition $V(x_1,...,x_{N-1},0)= V(-x_1,...,-x_{N-1},0) $ for all $x_j \in \R^3$   and $\kappa(-\cdot) = \kappa(\cdot)$,  then for all $\xi \in \R^3$ each eigenvalue of $H(\xi)$ has even or infinite multiplicity.
\end{theorem}

\begin{proof} The theorem follows as a consequence of  Parts (c) and (b) of  Lemma~\ref{lem:invt}, Theorem~\ref{le:kram}.
The first statement follows using the anti-linear symmetry $\mathcal{T}'$. The second statement follows using the anti-linear symmetry $\mathcal{P}' \mathcal{T}'$ and their commutativity
property, cf.  Proposition~\ref{prop:comm} and its trivial adaption to \eqref{defofHprime}.
\end{proof}

Next we consider quantum systems with identical particles. For notational simplicity, we shall assume that there is a single  particle which
is distinguishable from the rest. This is  satisfied for atoms, ions and many molecules. Otherwise,
a  further restriction  to subspaces would be necessary.

\begin{theorem}\label{thm:transinv2}   Suppose  $V$ is translationally invariant and   $\sum_{j=1}^N 2 s_j$ is odd.
Suppose that  the  partition $\mathfrak{P}$, the function $\tau$
 and the potential $V$,    satisfy Hypothesis  \ref{sym}. Furthermore, assume  $\{N \} \in \mathfrak{P}$ and let $\mathfrak{P}' =\mathfrak{P} \setminus \{ \{ N \} \}   $
and $\tau' = \tau|_{\mathfrak{P}'}$.
If $\overline{\kappa(\cdot)} = \kappa(-\cdot)$ each eigenvalue of $H(0)$ when restricted to  $  \HH'_{{\rm mat},\mathfrak{P}',\tau'} \otimes \mathcal{D}_{s_N}  \otimes \FF_s(\mathfrak{g}) $
has even or infinite multiplicity. If in addition  $V(x_1,...,x_{N-1},0)= V(-x_1,...,-x_{N-1},0) $ for all $x_j \in \R^3$ and $\kappa(-\cdot) = \kappa(\cdot)$,  then each eigenvalue of $H(\xi)$
when restricted to  $  \HH'_{{\rm mat},\mathfrak{P}',\tau'} \otimes  \mathcal{D}_{s_N} \otimes \FF_s(\mathfrak{g}) $
has even or infinite multiplicity.
\end{theorem}

\begin{proof} Follows from the same proof as  Theorem \ref{thm:transinv1},  by observing in addition that
$\mathcal{T}'$ and $\mathcal{P}'$ commute with
$\mathfrak{U}(\underline{\sigma})$
for any $\sigma \in \mathfrak{S}_p$ and $p \in \mathfrak{P}'$.
\end{proof}

\begin{remark}{\rm We note that the statement  of Theorem 6.2 was proven for
the special case where $N=1$ and $V=0$ for small coupling in \cite{HiroshimaSpohn.2001}
and for general coupling in \cite{Hiroshima.2007}.  Clearly, Theorem \ref{thm:transinv2}
covers the special case of $N-1$ electrons with spin 1/2 and a spinless
nucleus with  pairwise Coulomb interactions  ($\mathcal{P} = \{\{1,....,N-1\}, \{N\} \}  $),  cf. Remark 5.2 in   \cite{LosMiySpo09}.
We note that  whereas ground states of fiber Hamiltonians
describing electrons do not exist for nonzero momentum \cite{HasHer08-1}, they are shown to   exist
for  atoms and small absolute values of the  momentum \cite{LosMiySpo07}.}
\end{remark}

\section{Schr\"odinger Representation}  \label{sec:schr}

In this section we define rotation, parity and time reversal symmetry in the so called Schr\"odinger representation
of non-relativistic qed. To this end,
we recall the Schwartz space of smooth functions of rapid decrease  $\mathcal{S}(\R^d;\mathbb{F})$, with $\mathbb{F} = \R$ or  $\mathbb{F} = \C$, which is the set of infinitely
differentiable $\mathbb{F}$-valued functions $f(x)$ on $\R^d$ for which
\begin{equation} \label{eq:seminormschwartz}
\| f \|_{\alpha,\beta} = \sup_{x \in \R^d} |x^\alpha \partial^\beta f(x) | < \infty
\end{equation}
for all $\alpha, \beta \in \N_0^d$. Let $\underline{\mathcal{S}} = \mathcal{S}(\R^3;\R)^3$ equipped
with  the product topology.
The topological dual space $ \underline{\mathcal{S}}'$ can be identified with the set of all $T \in \mathcal{S}'(\R^3;\R)^3$, with $T(f) = T_1(f_1) + T_2(f_2) + T_3(f_3)$.

On $ \underline{\mathcal{S}} $ we  define the  symmetric positive semi-definite form
\begin{align} \label{defofB}
B( v , w )  = \sum_{i,j} \int  \frac{1}{|k|} \overline{\hat{v}_i(k) } P_{i,j}(k)    \hat{w}_j(k) d^3 k ,
\end{align}
where we recall
\begin{align}   \label{defofP2}
P(k)_{a,b} := \delta_{ab} - \frac{k_a k_b}{|k|^2}  , \quad a , b = 1,2,3 , \quad k \neq 0 .
\end{align}
Let
$$
c(f) = e^{-\frac{1}{4} B(f,f)}
$$
for $f \in \underline{\mathcal{S}}$.

By definition a {\it cylinder set} in $ \underline{\mathcal{S}}' $ is a set
$$
\{ T \in   \underline{\mathcal{S}}' : (T(f_1),....,T(f_n)) \in \Omega  \} ,
$$
where $f_1,...,f_n $ are $n$ fixed elements in $ \underline{\mathcal{S}}$ and $\Omega$ is a fixed Borel set in $\R^n$. A {\it cylinder set measure} on $  \underline{\mathcal{S}}' $  is a measure, $\mu$,
on the $\sigma$-algebra, generated by the cylinder sets, with  $\mu(  \underline{\mathcal{S}}'  ) = 1$. By construction, each $f \in  \underline{\mathcal{S}}$ defines a measurable function $\varphi(f)$ on
$  \underline{\mathcal{S}}'$ by
\begin{align} \label{varphidefofT}
\varphi(f)(T) = T(f).
\end{align}
In particular it follows that  for all $\alpha , \beta \in \R$ and $f, g \in  \underline{\mathcal{S}} $
\begin{align}
\varphi(\alpha f + \beta g) = \alpha \varphi(f) + \beta \varphi(g) . \label{eq:lin}
\end{align}
We shall use the following theorem, see \cite{GliJaf87,GliJaf85-1,GliJaf85-2,Fro77,FefFroGra97}.

\begin{theorem} \label{thm:mainschr1}  There exists a unique cylinder set measure $\nu$ on $ \underline{\mathcal{S}}' $ such that  for all $f \in  \underline{\mathcal{S}}  $
\begin{equation} \label{eq:fourierminlos}
\exp(- \frac{1}{4} B(f,f) ) = \int \exp ( i \varphi(f) ) d\nu
\end{equation}
Furthermore, $\nu$ has the following properties.
\begin{enumerate}[(a)]
\item For each $f \in    \underline{\mathcal{S}}  $ the function $\varphi(f)$ is a Gaussian random variable with mean zero and variance $\frac{1}{2} B(f,f)$.
\item   For  $f_1,...f_n \in  \underline{\mathcal{S}}   $ the random variables $\varphi(f_1),...,\varphi(f_n)$ are jointly Gaussian random variables.
\item  Let $\mathcal{U} = \{ F(\varphi(f_1),...,\varphi(f_n) ) : F \in \mathcal{S}(\R^n; \C) , f_1,...,f_n \in    \underline{\mathcal{S}}   \}$. Then $\mathcal{U}$ is dense in $L^2( \underline{\mathcal{S}}'  ,d\nu)$.
\item \label{supportofmeas}   If $f \in \underline{\mathcal{S}}$ and
$P \widehat{f} = 0$, then $\varphi(f) = 0$ almost surely, cf. \eqref{defofP2}.
 In particular, for almost all $T  = (T_1,T_2,T_3)  \in \underline{\mathcal{S}}'$ we have $\nabla \cdot T = 0$.
\end{enumerate}
\end{theorem}
A proof of  Theorem \ref{thm:mainschr1} will be given in Appendix~\ref{app:minlos}.
Henceforth, we shall denote by $\nu$ the unique measure on $\underline{\mathcal{S}}'$ satisfying  \eqref{eq:fourierminlos}.

\begin{remark} {\rm We note that part~\ref{supportofmeas} of Theorem   \ref{thm:mainschr1} will not be needed.
Nevertheless it is interesting in its own. }
\end{remark}

To formulate the next theorem we  define
$$
\underline{\mathcal{S}}_0 :=\{ g \in  \underline{\mathcal{S}} : \nabla \cdot g = 0 \}  .
$$
By $\overline{( \ \cdot \ )}^{\rm cl}$ we shall denote the operator closure.

\begin{theorem} \label{thm:unitransformexp}  There exists a unique  unitary transformation $V_\vv : \FF_s(\vv) \to L^2( \underline{\mathcal{S}}'   ,d\nu)$ with
the  following properties
\begin{enumerate}[(i)]
\item $ V_\vv \Omega = 1 $,
\item $V_\vv \overline{( a^*(i_\omega f) + a(i_\omega f ) )}^{\rm cl}  V_\vv^{-1} = \varphi(f) $,   for all $ f \in  \underline{\mathcal{S}}_0 $,
\end{enumerate}
where $i_\omega  f =  (\omega^{-1/2}  \hat{f})^{\vee} $ and $\varphi(f)$ is understood as a multiplication operator.
Moreover,  we  have  $V_\vv \Gamma( \mathcal{K}_{\vv}) = J  V_\vv$, where $J$ denotes complex conjugation in  $ L^2(  \underline{\mathcal{S}}'  ,d\nu)$. 
\end{theorem}

The proof of  Theorem \ref{thm:unitransformexp} will be given in Appendix  \ref{app:minlos}.  Using  Lemma \ref{lem:trafohtov} we obtain
immediately  the following corollary.

\begin{corollary}   \label{thm:unitransformexp2} Let the notation be as in  in  Theorem \ref{thm:unitransformexp}.
 There exists a unique  unitary transformation $V_{\gggg} : \FF_s(\gggg) \to L^2( \underline{\mathcal{S}}'  ,d\nu)$ with
the  following properties
\begin{enumerate}[(i)]
\item $  V_{\gggg} \Omega = 1 $,
\item $  V_{\gggg}  \overline{( a^*( \tau_\epsilon^{-1} i_\omega f) + a( \tau_\epsilon^{-1} i_\omega f ) )}^{\rm cl}  V_{\gggg}^{-1} = \varphi(f) $,   for all $ f \in \underline{\mathcal{S}}_0 $.
\end{enumerate}
Moreover,  we  have  $V_{\gggg} \Gamma(  \mathcal{K}_{\gggg} ) = J  V_{\gggg}$, where $J$ denotes complex conjugation in  $ L^2( \underline{\mathcal{S}}'   ,d\nu)$.
\end{corollary}

Next we  define  symmetries in Schr\"odinger representation.
We will show in Theorem~\ref{thm:symfockschr}, below,   that they  agree
by the unitary transformations  of  Theorem~\ref{thm:unitransformexp}  and Corollary~\ref{thm:unitransformexp2} with the  definitions in Fock space representation.
We define for  $U \in SU(2)$ on $ \underline{\mathcal{S}}$ the representation
$$
(\UU_{   \underline{\mathcal{S}}    }(U) f)(x) = R f(R^{-1} x )  ,  \quad f \in   \underline{\mathcal{S}}  , \ x \in \R^3 ,
$$
where $R = \pi(U)$.
We define for $f \in    \underline{\mathcal{S}}$
$$
(\mathcal{P}_{  \underline{\mathcal{S}}   } f )(x) = - f(-x) .
$$
As a consequence of the definition  $ \mathcal{P}_{  \underline{\mathcal{S}}}^{-1} = \mathcal{P}_{  \underline{\mathcal{S}} } $.
Then this defines by duality  a transformation on $  \underline{\mathcal{S}}'   $ by
$$
(\UU_{   \underline{\mathcal{S}}' }(U) T ) (f)  =  T  (\UU_{   \underline{\mathcal{S}}    }(U)^{-1} f  )
$$
and
$$
(\mathcal{P}_{  \underline{\mathcal{S}}'   }  T )(f) =  T(\mathcal{P}_{   \underline{\mathcal{S}}   }^{-1} f ) ,
$$
for all $T \in  \underline{\mathcal{S}}'   $ and $f \in   \underline{\mathcal{S}}$.
On $L^2( \underline{\mathcal{S}}'   ,d\nu)$  we define for any $F \in L^2( \underline{\mathcal{S}}'    ,d\nu)$
\begin{align*}
(\UU_{ \rm Sch }(U) F )(T)  &  =   F( \UU_{\underline{\mathcal{S}}'}(U)^{-1}  T)  , \quad U \in SU(2) ,   \\
(\mathcal{P}_{  \rm Sch  } F)(T)  & = F(\mathcal{P}_{\underline{\mathcal{S}}'}^{-1}  T)  ,  \\
(\mathcal{K}_{  \rm Sch  } F )(T) & = \overline{F(T)} \\
 (\Theta_{\rm Sch} F)(T)  & = F(-T)
\end{align*}
for all $T \in   \underline{\mathcal{S}}'$.

\begin{lemma} Let $U \in SU(2)$.  The measure $\nu$ is  invariant with respect to  $\UU_{   \underline{\mathcal{S}}'   }(U)$ and  $\mathcal{P}_{   \underline{\mathcal{S}}'  }$. The  transformations $\UU_{  \rm Sch }(U)$,  $\mathcal{P}_{   \rm Sch  }$  are unitary
transformations on  $L^2(  \underline{\mathcal{S}}'   ,d\nu)$. The transformation $\mathcal{K}_{\rm Sch  }$ is an anti-unitary transformation  on $L^2(  \underline{\mathcal{S}}'   ,d\nu)$, which squares to one.
The measure $\nu$ is invariant with respect to $-1_{\underline{\mathcal{S}}'}$, and $\Theta_{\rm Sch}$ is a unitary
transformation on $L^2(\mathcal{S}',d\nu)$, which squares to one.
\end{lemma}
\begin{proof} Let $G$ stand for $\UU_{  \underline{\mathcal{S}}'   }(U)$ and
$\mathcal{P}_{   \underline{\mathcal{S}}' }$ and $g$ for
$\UU_{  \underline{\mathcal{S}}   }(U)$ and $\mathcal{P}_{   \underline{\mathcal{S}} }$,
respectively.
Then $G$   leaves the set of cylinder sets invariant, and hence the $\sigma$-algebra
generated by the cylinder sets. Since the form $B$ is invariant with
respect to  $G$, so is the measure $\nu$. To see this define $\nu_G(A) = \nu(G(A))$ for any measurable set $A$.
 Then for any $f \in    \underline{\mathcal{S}} $ we find from the definition of the integral
\begin{align*}
\exp(-\frac{1}{4}B(  f,  f)) & = \exp(-\frac{1}{4}B( g f,  g  f) ) =  \int \exp(i \varphi( g f) ) d \nu =   \int \exp(i  (G^{-1}  T)(  f) ) d \nu(T)   \\
&  =  \int \exp(i   T(  f) ) d \nu_G(T) =  \int \exp(i   \varphi(  f) ) d \nu_G .
\end{align*}
Thus it  follows $\nu = \nu_G$ from the uniqueness property  in  Theorem  \ref{thm:mainschr1}. Thus the  unitarity properties of $\mathcal{U}_{\rm Sch}(U)$
and $\mathcal{P}_{\rm Sch}$  on $L^2(  \underline{\mathcal{S}}'  ,d\nu)$
now follow  by the definition of the integral as a limit of simple functions.  The anti-unitarity of  $\mathcal{K}_{ \rm Sch   }$ is obvious.
The last statement about $\Theta_{\rm Sch}$ follows analogously  as above with $G =
- 1_{\underline{\mathcal{S}}'}$
and $g = - 1_{\underline{\mathcal{S}}}$.
\end{proof}

The following theorem relates the symmetries  in  the  Fock representation to
the symmetries in the Schr\"odinger representation.

\begin{theorem} \label{thm:symfockschr} Let  $V_{\vv}$ and $V_{\gggg}$ be the unique unitary  transformations  satisfying  (i) and  (ii) of Theorem  \ref{thm:unitransformexp} and Corollary  \ref{thm:unitransformexp2}, respectively.
Then the following identities hold.
\begin{enumerate}[(a)]
\item   $V_{\gggg}  \UU_{\rm f}(U) V_{\gggg}^{-1} = \UU_{ \rm Sch  }(U)$ and  $V_{\vv}  \Gamma(\UU_\vv(\pi(U)) V_{\vv}^{-1} = \UU_{\rm Sch}(U)$, for $U \in SU(2)$,
\item $V_{\gggg}  \PP_{\rm f} V_{\gggg}^{-1} = \PP_{  \rm Sch  }$ and  $V_{\vv}  \Gamma(\mathcal{P}_\vv)  V_{\vv}^{-1} = \PP_{\rm Sch}$,
\item   $V_{\gggg}  \Gamma(\mathcal{K}_\gggg ) V_{\gggg}^{-1} = \mathcal{K}_{  \rm Sch }$ and $V_{\vv}  \Gamma( \mathcal{K}_{\vv}) V_\vv^{-1} =  \mathcal{K}_{ \rm Sch    } $.
\item $V_{\gggg}  \Gamma( - 1_\gggg ) V_{\gggg}^{-1} = \Theta_{  \rm Sch }$ and
$V_{\vv}  \Gamma( - 1_\vv ) V_\vv^{-1} =  \Theta_{ \rm Sch    } $.
\item $V_{\gggg}  \TT_{\rm f} V_{\gggg}^{-1} =  \Theta_{ \rm Sch    } \mathcal{K}_{ \rm Sch    }$.
\end{enumerate}
\end{theorem}
\begin{proof} We only discuss  the case for $\vv$, the case for $\gggg$ then follows using Lemma \ref{lem:trafohtov}. \\
 (a) Let $W = \UU_{  \rm Sch }(U)  V_{\vv} \Gamma( \UU_\vv(\pi(U))^{-1})$.
Then it follows from the definitions that   $W \Omega = 1$. Furthermore, it follows  for all $f \in   \underline{\mathcal{S}}_0$ using \eqref{eq:GammaasharpGamma},
the invariance of $\omega$ and  Theorem   \ref{thm:unitransformexp} (ii)
\begin{align}
W (\overline{a^*(i_\omega f) + a(i_\omega f)})^{\rm cl} W^{-1} & = \UU_{ \rm Sch   }(U) V_\vv  (\overline{a^*(i_\omega \UU_\vv(U)^{-1} f) + a(i_\omega \UU_\vv(U)^{-1}  f)})^{\rm cl}  V_\vv^{-1} \UU_{  \rm Sch }(U)^{-1} \nonumber  \\
& =\UU_{\rm Sch }(U)\varphi(\UU_\vv(U)^{-1} f) \UU_{ \rm Sch  }(U)^{-1} = \varphi(f),  \label{didentW}
\end{align}
where the last equality can be seen as follows.  For any $F \in L^2(\underline{\mathcal{S}}', d\nu)$ we find with $F' :=  \UU_{ \rm Sch  }(U)^{-1} F$ using
 $\UU_\vv(U) f = \UU_{  \underline{\mathcal{S}}   }(U) f$ and  inserting into the definitions, e.g. \eqref{varphidefofT},  that
\begin{align*}
& (\UU_{\rm Sch }(U) ( \varphi(\UU_\vv(U)^{-1} f)  F' ))(T)  =  (\varphi(\UU_{\underline{\mathcal{S}}}(U)^{-1} f)  F')( \mathcal{U}_{\underline{\mathcal{S}}'}(U)^{-1}T) \\
& =  ( \mathcal{U}_{\underline{\mathcal{S}}'}(U)^{-1}T) (\UU_{\underline{\mathcal{S}}}(U)^{-1} f) F'( \mathcal{U}_{\underline{\mathcal{S}}'}(U)^{-1}T) \\
&=  T(f) F'(  \mathcal{U}_{\underline{\mathcal{S}}'}(U)^{-1}T)  = \varphi(f) F( T) .
\end{align*}
This show the last equality in \eqref{didentW}.
It now follows from   \eqref{didentW}  that  $W = V_\vv$ by the uniqueness statement  of   Theorem  \ref{thm:unitransformexp}. This shows (a). Now
(b) is  shown similarly as~{(a).}

\noindent
(c) Let $W = \KK_{  \rm Sch }  V_{\vv} \Gamma(  \KK_\vv)$.
Then it follows from the definitions that   $W \Omega = 1$. Furthermore, it follows  for all $f \in   \underline{\mathcal{S}}_0$ using \eqref{eq:GammaasharpGamma},
the reality and symmetry assumptions of $\omega$, and  Theorem   \ref{thm:unitransformexp} (ii) that
\begin{align*}
W (\overline{a^*(i_\omega f) + a(i_\omega f)})^{\rm cl} W^{-1} & = \KK_{ \rm Sch   } V_\vv  (\overline{a^*( \KK_\vv i_\omega  f) + a( \KK_\vv i_\omega   f)})^{\rm cl}  V_\vv^{-1} \KK_{  \rm Sch }^{-1}   \\
& = \KK_{ \rm Sch   } V_\vv  (\overline{a^*( i_\omega  \KK_\vv  f) + a( i_\omega \KK_\vv  f)})^{\rm cl}  V_\vv^{-1} \KK_{  \rm Sch }^{-1}  \\
& =\KK_{\rm Sch }\varphi(  \KK_\vv f) \KK_{ \rm Sch  }^{-1} = \varphi(f) . 
\end{align*}
As in (a) it now follows   that  $W = V_\vv$ by the uniqueness statement  of   Theorem  \ref{thm:unitransformexp}. This shows (c), since $\mathcal{K}_\vv^{-1} = \mathcal{K}_\vv$. Now (d) follows analogously to (c) by  considering
$W = \Theta_{  \rm Sch }  V_{\vv} \Gamma(  - 1_\vv)$ and observing that
\begin{align*}
W (\overline{a^*(i_\omega f) + a(i_\omega f)})^{\rm cl} W^{-1} & = \Theta_{ \rm Sch   } V_\vv  (\overline{a^*( -  i_\omega  f) + a( -  i_\omega   f)})^{\rm cl}  V_\vv^{-1} \Theta_{  \rm Sch }^{-1}  \\
& = \Theta_{ \rm Sch   } V_\vv  (\overline{a^*( i_\omega   (- f )) + a( i_\omega  (- f)})^{\rm cl}  V_\vv^{-1} \Theta_{  \rm Sch }^{-1}   \\
& =\Theta_{\rm Sch }\varphi(  -  f) \Theta_{ \rm Sch  }^{-1} = \varphi(f) .  
\end{align*}
Again by uniqueness $W = V_\vv$. This shows (d).
Finally, (e) follows from (c) and (d).
\end{proof}

\begin{remark} {\rm  We see from Subsection \ref{subsecteim}   and Theorem \ref{thm:symfockschr}
that $\UU_{\rm Sch}$, $\PP_{\rm Sch}$ and $\TT_{\rm Sch } :=  \Theta_{\rm Sch}    \mathcal{K}_{\rm Sch} $ correspond to  the rotation, parity and time reversal symmetries in
 the Schr\"odinger representation.
Alternatively, one could redefine the field operators in the Hamiltonian so that $\mathcal{K}_{\rm Sch}$ has the  property of a   time reversal symmetry, cf.  \cite{LosMiySpo10}.
}
\end{remark}

\section*{Acknowledgements}

\noindent Both authors acknowledge financial support by the Research Training Group (1523/2) “Quantum and Gravitational Fields” when this project was initiated.
D. Hasler wants to thank I. Herbst for valuable discussions on the subject.
M. Lange  also acknowledges financial support from the European Research Council (ERC) under the European Union’s Horizon 2020 research and innovation programme (ERC StG MaMBoQ, grant agreement n.802901).

\appendix

\section{Gaussian Random Processes}\label{app:GaussianRandom}

In this appendix we review notations and results  about so called
Gaussian random processes. We follow \cite{sim74}. The main result is Theorem~\ref{thm:mainindex}, which will be used in the proof of Theorem~\ref{thm:unitransformexp} in Appendix~\ref{app:minlos}.
First we introduce the following definitions.

\begin{definition}  Let $(M,\mu)$ be a  probability measure space.  Let $V$ be a real vector space.
A {\bf random process indexed by} $V$ is a map $\phi$ from $V$ to the random variables on $M$, so that
almost everywhere
\begin{align*}
\phi(v + w) &= \phi(v) + \phi(w)  \quad \forall v , w \in V  \\
\phi(\alpha v ) &= \alpha \phi(v) \quad \forall \alpha\in \R , \forall v \in V .
\end{align*}
\end{definition}

For a random variable $Y$ on  probability measure space $(M,\mu)$ we will use the notation
 $$\langle Y \rangle :=  \int Y d \mu. $$

\begin{definition} \label{def:gaussrandproc}
Let $\rr$ be a real Hilbert space with inner product $\langle \cdot , \cdot \rangle_\rr$. A {\bf Gaussian random process indexed by} $\rr$
is a random process $\phi$  indexed by $\rr$   so that the following holds.
\begin{enumerate}[(a)]
\item  The set $\{ F(\phi(v_1),...,\phi(v_n) ) : v_1,...,v_n \in \rr , F \in \mathcal{S}(\R^n) \}$ is dense
in $L^2(M, d\mu)$, where $(M,\mu)$ is the probability  measure space of the random process $\phi$.
\item Each $\phi(v)$ is a Gaussian random variable.
\item $ \inn{\phi(v) \phi(w) } = \frac{1}{2} \inn{ v , w }_\rr$.
\end{enumerate}
\end{definition}

\begin{remark} \quad {\rm
\begin{enumerate}[(a)]
\item We note that in (a)  of Definition \ref{def:gaussrandproc}, we use
a different assumption than  in the definition of a Gaussian random process indexed by a Hilbert space in \cite{sim74}. However, in view of \cite[Lemma~I.5]{sim74}
this is equivalent.
\item  One can show that two Gaussian random processes indexed by the same real
Hilbert space  are  unique up to isomorphisms of probability measure spaces, see for example \cite[Theorem~I.6]{sim74}.
\item For any real Hilbert space $\rr$, a  Gaussian process indexed by $\rr$ exists. For a proof see Theorem~I.9 in \cite{sim74}.
\end{enumerate}}
\end{remark}

Let  $\rr $ be the complexification of $\rr$, i.e., $\rr_\C = \rr \oplus \rr$ as a real Hilbert space with a complex
structure given by $i(u,v) = (- v, u)$. We define \begin{align} \label{deofofJ}
J : \rr_\C \to \rr_\C , \qquad J (u,v) = (u,-v). \end{align}  Then $J$ is anti-linear and satisfies $J^2=1$.
Without mention we shall imbed $\rr$ in $\rr_\C$ by the map $\iota :  u  \mapsto (u,0)$.  For the operator introduced in  \eqref{defastar}
we shall write for notational convenience
$a^\#(f) = a^\#(\iota f )$ for $f \in \rr$.

Next we introduce the notion of Wick powers and Wick product of random variables.
To this end we introduce the following  multi-index notation. For $k \in \N$, $\underline{n} \in \N_0^k$ and $\alpha , \beta \in \C^k$
we define
$$
\alpha^{\underline{n}} = \prod_{j=1}^k \alpha^{n_j} , \quad \alpha \beta = \sum_{j=1}^k \alpha_j \beta_j  , \quad |\underline{n}| = \sum_{j=1}^k n_j  , \quad \underline{n}! = \prod_{j=1}^k n_j! \quad .
$$
Given  a formal power series in random variables $f_1, ...,f_k$  with finite moments on a measure space $(M,\mu)$,
which  we denote by $\sum_{\underline{n} \in \N_0^k } a_{\underline{n} }f^{\underline{n}}, $ where $a_{\underline{n}}  \in \C$
and  $$f^{\underline{n}} := \prod_{j=1}^k f_j^{n_j} , $$
we  define
the formal derivative
\begin{align*}
\frac{\partial}{\partial f_i}   \sum_{\underline{n} \in \N_0^k } a_{\underline{n} }f^{\underline{n}}  =
\sum_{\underline{n} \in \N_0^k } a_{\underline{n}} n_i f^{\underline{n}-\underline{e}_i }  .
\end{align*}
where $\underline{e}_i \in \N^k$ is defined such  that all components vanish except the $i$-th, which equals 1.

\begin{remark}{\rm  As in \cite{sim74}   we don't identify two series which are identical by virtue of substituting in
specific  arguments
(e.g. $f$ and $f^2$ are distinct as formal power series even if $f=1$). }
\end{remark}

\begin{definition} \label{def:wickorder}
Let $f_1, ...,f_k$ be random variables with finite moments on a measure space $(M,\mu)$. The {\bf Wick product}
$: f^{\underline{n}} :$ is defined inductively in $n =|\underline{n}| $ by
\begin{enumerate}[(i)]
\item $:f^{\underline{0}}: \ = 1 $, where $\underline{0} = (0,....,0)$,
\item $ \langle :f^{\underline{n} } : \rangle = 0  $ if $n\neq 0$,  
\item  $\frac{\partial}{\partial f_i} :  f^{\underline{n} }: \ = n_i : f^{\underline{n}- \underline{e}_i } :$ \quad .
\end{enumerate}
\end{definition}

The following theorem is the main theorem of this section.

\begin{theorem} \label{thm:mainindex}  Let $\phi$ be a Gaussian random process indexed by a separable real Hilbert space  $\rr$ on
the probability measure space $(M,\mu)$, and let
 $D$ be  a dense subset of $\rr$.  Then  there exists a unique unitary transformation $V : \FF_s(\rr_\C) \to L^2(M,d\mu)$ satisfying
\begin{enumerate}[(i)]
\item $V \Omega = 1 $
\item $V (\overline{a^*(f) + a(f)})^{\rm cl} V^{-1} = \phi(f) $ for all $f \in D$.
\end{enumerate}
Moreover, the following holds. We have
\begin{enumerate}[(a)]
\item  $V (\overline{a^*(f) + a(f)}) ^{\rm cl} V^{-1} = \phi(f) $ holds for all   $f \in \rr$.
\item  $\mathcal{J} V = V \Gamma(J)$, where $J$ is defined in   \eqref{deofofJ}     and $\mathcal{J}$ denotes ordinary complex  conjugation in $L^2(M,d\mu)$.
\item For all $f_j \in \rr$  we have \begin{equation} \label{constrofD}
V a^*(f_1) \cdots a^*(f_n) \Omega = \quad  :\phi(f_1) \cdots \phi(f_n): . 
\end{equation} 
\end{enumerate}
\end{theorem}
A proof of Theorem \ref{thm:mainindex} can be found  in  Theorems I.6 and  I.11 in  \cite{sim74}. For the convenience
of the reader,  we shall outline a proof below.
First,  we need a few lemmas.
For  random variables $f_1,\cdots ,f_k$ with finite moments we define the formal power series for $\alpha  \in \C^k$ by
\begin{equation} \label{eq:defofexp}
: \exp( \alpha f ) : \ = \sum_{\underline{n} \in \N_0^k }^\infty \frac{  \alpha^{\underline{n}}   :f^{\underline{n}}:}{\underline{n}!} .
\end{equation} 

\begin{lemma} \label{lem:fundrelwick}  Let  $f_1,...,f_k$ be  random variables with finite moments on
a probability measure space $(M,\mu)$. Then
for all $\alpha \in \C^k$ the following holds
\begin{enumerate}[(a)]
\item $\langle : \exp(\alpha f ) : \rangle = 1 $
\item $: \exp(\alpha f ) : = \exp(\alpha f ) \langle \exp(\alpha f ) \rangle^{-1}$
\item  If $f$ is a Gaussian random variable, then  \eqref{eq:defofexp} converges in $L^1(M,d\mu)$
and  \label{lem:fundrelwick1}   $$: \exp( \alpha f ): = \exp( \alpha f ) \exp\left(-\frac{1}{2} \sum_{i,j} \alpha_i \alpha_j   \langle f_i f_j  \rangle \right) .  $$
\end{enumerate}
\end{lemma}
\begin{proof}
(a) This follows from (i) and  (ii) of Definition \ref{def:wickorder}.   (b) By  (iii) of Definition \ref{def:wickorder},
we find $\frac{\partial}{\partial f_j } : \exp(\alpha f ):  =\alpha_j : \exp( \alpha f ) : $. Thus
$\frac{\partial}{\partial f_j }  :\exp(\alpha f) : \exp(-\alpha f) = 0$ and so  $ :\exp(\alpha f) :  \exp(-\alpha f) =C$
for some constant  $C$. Thus from (a) it follows that $C = \langle \exp(\alpha f ) \rangle^{-1}$.  (c) The $L^1$ convergence follows from dominated convergence. Using that $f$
is    Gaussian  one finds   $\langle \exp(\alpha f ) \rangle = \exp(\frac{1}{2} \sum_{i,j}  \alpha_i \alpha_j  \langle f_i f_j  \rangle)$ (e.g. by calculating the Fourier transform for $\alpha = i t$, with $t \in \R^k$,  and then using analytic continuation).   Thus (c) follows from (b).
\end{proof}

The following Lemma is from \cite[Theorem I.3, Corollary I.4]{sim74}.

\begin{lemma} \label{lem:orthogrel}
The following holds.
\begin{enumerate}[(a)]
\item  If  $f$ and $g$ are  Gaussian random variables,  then for $m,n \in \N_0$
\begin{align*}
\langle : f^n :  : g^m : \rangle =  \delta_{{n}, {m} } {n}! \langle f g \rangle^n  .
\end{align*}
\item  \label{lem:orthogrel1} If $f_1,...,f_n$ and $g_1,...,g_m$ are Gaussian random variables and $n \neq m$, then
\begin{align*}
\langle : f_1 \cdots f_n :  : g_1  \cdots g_m : \rangle =  0 .
\end{align*}
\item  \label{lem:orthogrel2}   If $f_1,...,f_k$ are Gaussian random variables with $\langle f_i f_j \rangle = \delta_{i,j}$, then for $\underline{n}, \underline{m} \in \N_0^k$
\begin{align*}
\langle : f^{\underline{n}}  :  :   f^{\underline{m}}  : \rangle =  \delta_{\underline{n}, \underline{m} } \underline{n}!  \ .
\end{align*}
\end{enumerate}
\end{lemma}
\begin{proof} (a).  By  \ref{lem:fundrelwick1}  of Lemma  \ref{lem:fundrelwick}  we find
\begin{align*}
: \exp( \alpha f ) : : \exp( \beta g ) :  & = \exp(\alpha f + \beta g )
 \exp \left( - \frac{1}{2} \left[ \alpha^2 \langle f^2 \rangle + \beta^2  \langle g^2 \rangle \right] \right)  \\
 & = : \exp(\alpha f + \beta g ) :
 \exp \left( \alpha \beta \langle f g \rangle  \right)  .
\end{align*}
Thus by (a)   of Lemma  \ref{lem:fundrelwick}
\begin{align*}
 & \langle : \exp( \alpha f ) : : \exp( \beta g ) : \rangle  =
 \exp \left(   \alpha \beta \langle f g \rangle  \right)  .
\end{align*}
Thus  (a) now  follows  by expanding exponentials  and equating coefficients.  (b,c) follow from the multinomial theorem and (a).
\end{proof}

\begin{lemma} \label{lemcompleteranspace}  Let $\phi$ be a Gaussian random process indexed by the real Hilbert space $\rr$. Let $$\Gamma_n(\rr) =
 \overline{ {\rm lin}_\C  \{:\phi(f_1) \cdots \phi(f_n): \quad  |  \quad f_1,...,f_n \in \rr  \} }^{\rm cl}, \quad n \in \N$$ and $\Gamma_0(\rr) = \C$.
Then  the following holds.
\begin{enumerate}[(a)]
\item $\Gamma_n(\rr) \perp \Gamma_m(\rr)$  for $n \neq m$.
\item
$L^2(M,d\mu) = \bigoplus_{n=0}^\infty \Gamma_n(\rr)  . $
\end{enumerate}
\end{lemma}
\begin{proof}
(a) This follows from \ref{lem:orthogrel1}  of Lemma~\ref{lem:orthogrel}.  (b) For any $f \in \rr$,
a direct computation shows that the formal power series $: e^{ i \phi(f) }:$ converges in $L^2(M,d\mu)$.
We shall denote the limit by the same symbol. Thus by definition $\bigoplus_{n=0}^\infty \Gamma_n(\rr)$
contains  $: e^{ i \phi(f)}: $ and so $ e^{ i \phi(f)} $  in view of  \ref{lem:fundrelwick1} of
Lemma  \ref{lem:fundrelwick}. In particular, for any $F \in \mathcal{S}(\R^n)$ and $f_1,...,f_n \in \rr$ we find that
\begin{equation}  \label{eq:defofF}
F(\phi(f_1),\cdots,\phi(f_n) ) = (2\pi)^{-n/2} \int \widehat{F}(t) \exp(\sum_{j=1}^n t_j \phi(f_j) ) d^n t
\end{equation}
is in $\bigoplus_{n=0}^\infty \Gamma_n(\rr)$. But  the set of  random variables of the form as  on the
left hand side of   \eqref{eq:defofF}   are dense in $L^2(M,d\mu)$ by  the  assumptions of an indexed  Gaussian random process.
Thus (b) follows.
\end{proof}

\begin{proof}[Proof of Theorem \ref{thm:mainindex}]
First we show uniqueness. To this end  we define for $f \in \mathfrak{r}_\C$  the operator $\phi_\FF(f)$ in $\FF_s(\rr_\C) $    by
\begin{align} \label{defofphiFF} \phi_\FF(f) =\overline{ a^*(f) + a(f)}^{\rm cl} .  \end{align}
We claim that  for any $m \in \N_0$ the set
$$
\{ \phi_\FF(f_1) \cdots \phi_\FF(f_n) \Omega : f_i \in D , n=0,1,... , m\}
$$ 
is dense in $\bigoplus_{n=0}^m {S}_n(  \rr_\C^{\otimes n} )$.
To show this,  we use induction in $m$.  The claim clearly holds for $m=0$.
Suppose it holds for $m$. Then multiplying out,  we find $$\phi_\FF(f_1) \cdots \phi_\FF(f_{m+1})\Omega = a^*(f_1) \cdots a^*(f_{m+1}) \Omega
+ h,$$ where $h \in \bigoplus_{n=0}^m {S}_n(  \rr_\C^{\otimes n} )$. Since  the linear span of  $a^*(f_1) \cdots a^*(f_{m+1}) \Omega $
is dense in  $ {S}_{m+1}(  \rr_\C^{\otimes n} )$ the claim follows for $m+1$.  Since
$$
V \phi_\FF(f_1) \cdots \phi_\FF(f_n) \Omega = ( V  \phi_\FF(f_1) V^{-1} )   \cdots  ( V \phi_\FF(f_n) V^{-1})  V \Omega , 
$$  
properties (i) and (ii) determine the action of $V$ uniquely on a dense set.\\
Let us now show existence. First  choose an o.n.b. $\mathcal{B}$ of $\rr$.
Define  $V$ by $V \Omega =1$ and
\begin{align*}
V a^*(f_1) \cdots a^*(f_n) \Omega = \quad  :\phi(f_1) \cdots \phi(f_n):  ,
\end{align*}
where $f_j \in \mathcal{B}$
(this is well defined by the symmetry property of the Wick product) and extend it  by linearity.
It is straight forward to see that the
map  $V$ is an isometry using on the one hand side the canonical commutation relations for creation and annihilation operators in Fock space
  and on the other hand   Lemma  \ref{lem:orthogrel}. Surjectivity, and hence unitarity,  follows from  Lemma \ref{lemcompleteranspace}. 
Obviously,  $V$ satisfies (i)   by construction. Let us now  show, that it satisfies (a) and hence  (ii).
 Using the definition,    \eqref{defofphiFF},
and the canonical commutation relations we find  for  $f_j \in \mathcal{B}$
\begin{align} \label{eq:phirel1}
& \phi_\FF(f_1) a^*(f_1)^{n_1}  \cdots a^*(f_k)^{n_k}  \Omega\\
& \qquad    = a^*(f_1)^{n_1+1}  \cdots a^*(f_k)^{n_k}  \Omega + n_1 a^*(f_1)^{n_1-1}  \cdots a^*(f_k)^{n_k}  \Omega .  \nonumber
\end{align}
On the other hand  we will show  that
\begin{align} \label{eq:phirel2}
& \phi(f_1)  : \phi(f_1)^{n_1}  \cdots \phi(f_k)^{n_k}  : \\
& \qquad    = \ : \phi(f_1)^{n_1+1}  \cdots \phi(f_k)^{n_k} :    + n_1  : \phi(f_1)^{n_1-1}  \cdots \phi(f_k)^{n_k} : . \nonumber
\end{align}
To see  \eqref{eq:phirel2}, we  first note that using  \ref{lem:fundrelwick1}  of Lemma  \ref{lem:fundrelwick} we obtain
\begin{align} \label{eq:phirel3}
& \phi(f_1)  : \exp(\sum_{j=1}^n \alpha_j \phi(f_j) )  :  \  = \left( \frac{\partial}{\partial \alpha_1} + \alpha_1 \right)  : \exp(\sum_{j=1}^n \alpha_j \phi(f_j) )  :  .
\end{align}
Now expanding   \eqref{eq:phirel3}  in a power series, calculating the derivative, and   equating coefficients,
we obtain   \eqref{eq:phirel2}.   Thus  it follows    in view of \eqref{eq:phirel1},   \eqref{eq:phirel2}, and  from the definition
of $V$  that for all $f \in \mathcal{B} $
\begin{align}
V \overline{( a^*(f) + a^*(f))}^{\rm cl} V^{-1} = \phi(f) .
\end{align}
This implies (a) (and hence (ii))  by linearity and continuity.
Clearly, (c) follows from uniqueness of   the above construction and multi-linearity.
To show  (b) observe that from  \eqref{constrofD} we find for any $f_j \in \rr$ that
 \begin{align*}
&\mathcal{J}  V  \Gamma(J)     a^*(f_1) \cdots a^*(f_n)  \Omega  = \mathcal{J}   V     a^*( J f_1) \cdots a^*( J f_n) ) \Omega  \\ & = \mathcal{J}  V     a^*(  f_1) \cdots a^*(  f_n) ) \Omega  = \mathcal{J}  :\phi(f_1) \cdots \phi(f_n): \\
& = \  \ :\phi(f_1) \cdots \phi(f_n): \  \ = V   a^*(f_1) \cdots a^*(f_n)  \Omega.
\end{align*}
Thus by density and $\C$-linearity it follows that $\mathcal{J} V \Gamma(J) = V$. Thus (b) follows, since $J^{-1} = J$.
\end{proof}

\section{An Application of Minlos' theorem}
\label{app:minlos}

In this appendix we will prove Theorems \ref{thm:mainschr1}
and  \ref{thm:unitransformexp}.
For this we shall introduce  the following definitions from \cite{sim74}.
Let us first recall the definition
$$
c(f) = e^{-\frac{1}{4} B(f,f)}
$$
for $f \in \underline{\mathcal{S}}$ with $B$ defined in  \eqref{defofB}.

\begin{lemma} \label{lem:positivity}  The following holds.
\begin{itemize}
\item[(i)]  $c(0)=1$.
\item[(ii)] $f  \mapsto  c(f)$ is continuous.
\item[(iii)] For any $f_1,....,f_n \in   \underline{\mathcal{S}} $ and $z_1,...,z_n \in \C$ we have
$$
\sum_{i,j=1}^n z_i \overline{z}_j  c(f_i-f_j)  \geq 0  \quad .
$$
\end{itemize}
\end{lemma}
\begin{proof} (i) This follows from $B(0,0) = 0$.
(ii)   It is straight forward to see that $f \mapsto B(f,f)$ is continuous on $    \underline{\mathcal{S}} $,
and hence also   the function $c : f  \mapsto  \exp(-\frac{1}{4} B(f,f))$. (iii)
Let $V = {\rm lin}_\R \{ f_1,...,f_n\}$. Then there exists a basis $(e_j)_{j=1,...,m}$
of $V$, with dual basis $(b_j)_{j=1}^m$, such that
$B(e_i,e_j) = \lambda_i \delta_{i,j}$ with $\lambda_1=\cdots = \lambda_p=1$ and $\lambda_{p+1}=\cdots \lambda_m=0$ for some $1 \leq p \leq m$.  Using that the Fourier transform
of a Gaussian is a Gaussian we find  for any $f \in V$  with $ f_j = b_j(f)$
$$
c(f)  = e^{-\frac{1}{4} B(f,f)} = e^{-\frac{1}{4} \sum_{j=1}^p f_j^2 }
= (\pi)^{-p/2}\int e^{ - i \sum_{j=}^p  y_j b_j(f)  } e^{- \sum_{j=1}^p y_j^2} d^py .
$$
So positivity of $c(f)$ now follows from Bochner's theorem \cite[Theorem IX.9]{ReeSim2}.
\end{proof}

\begin{proof}[Proof of Theorem  \ref{thm:mainschr1}]  The existence and uniqueness of the measure $\nu$   follows
in view of Lemma \ref{lem:positivity}
from  Minlos  theorem \cite[Theorem 3.4.2]{GliJaf87} see also \cite{Min63,Ume65,Bou69,Bog07,BogSmo17}.  To this end, we  extend the seminorms \eqref{eq:seminormschwartz} to $  \underline{\mathcal{S}}  $ as follows.  For $f = (f_1,f_2,f_3) \in  \underline{\mathcal{S}}$ we define $\| f \|_{\alpha,\beta} := \| f_1 \|_{\alpha,\beta}  + \| f_2 \|_{\alpha,\beta} + \| f_3 \|_{\alpha,\beta}$.
Then it  is straight forward to see that $ \underline{\mathcal{S}}$  with these seminorms
is a nuclear space.  \\
 \noindent (a) This follows since for $f \in \underline{\mathcal{S}}$ and each $t \in \R$ we have by  \eqref{eq:fourierminlos}
$$
\int \exp( i t \varphi(f)) d\nu = \exp(-\frac{1}{4} t^2 B(f,f) ) ,
$$
and so $\varphi(f)$  is a Gaussian random variable with mean zero, see \cite{sim74}.
\noindent
(b) This follows from (a) and linearity  \eqref{eq:lin}, see \cite{sim74}.
\noindent
(c)  We argue similarly as in  \cite{Ume65}.
 First observe that for all  measurable sets $E$ we have
\begin{equation} \label{densofset}
\forall \epsilon > 0 , \exists C  \text{ a cylinder set} , \quad  \nu(C \triangle E) < \epsilon .
\end{equation}
 Here, $\triangle$ stands for the symmetric
difference.  To this end,  let $\mathcal{E}$ be the set of all measurable  $E$  which
satisfy  \eqref{densofset}.  It is straight forward to verify that  $\mathcal{E}$ is a $\sigma$-algebra containing
all cylinder sets. Hence $\mathcal{E}$ equals the set of all measurable sets.  It follows by definition of the integral  that $ \{ 1_\Omega(\varphi(f_1),...,\varphi(f_n) ) : \Omega \subset \R^n \text{ Borel measurable}, f_1,...,f_n \in  \underline{\mathcal{S}} \}$ is dense in $L^2( \underline{\mathcal{S}}',d\nu)$. Now  it is well known that  $\mathcal{S}(\R^n ; \C)$
is dense in $L^1(\R^n , d\mu_C)$, where $\mu_C$ denotes Gaussian measure with covariance $C$ (with  possibly matrix elements  which are infinite).  This shows the density.
\noindent
(d)  If $f \in  \underline{\mathcal{S}}$ with $P \widehat{f} =0$, then $B(f,f) = 0$, so $\varphi(f)$ is by (a) a Gaussian random variable with variance zero.
Thus for all $f \in  \underline{\mathcal{S}}$ with $P \widehat{ f} = 0$ it follows that $\varphi(f) = 0$ almost everywhere. Now let  $h \in \mathcal{S}(\R^3;\R)$.
Then  $\nabla h \in \underline{\mathcal{S}}$ and for all $T \in  \underline{\mathcal{S}}'$ we have
\begin{align*}
 \varphi(\nabla h )(T) = 0 \Leftrightarrow  T( \nabla  h ) = 0  \Leftrightarrow
(\nabla \cdot  T)(h) = 0  .
\end{align*}
Since  $P \widehat{ \nabla h} = 0$,   we find   $(\nabla \cdot T )(h) = 0$ for  almost  all $T \in  \underline{\mathcal{S}}'$.
Since $\mathcal{S}(\R^3;\R)$ is separable, there exists a countable dense subset $\mathcal{Q}$.
It follows that for almost all $T \in \underline{\mathcal{S}}'$ we have $(\nabla \cdot  T)(h)=0$   for all $h \in \mathcal{Q}$.
Since  $\nabla \cdot T$ is continuous it follows that for almost all  $T \in  \underline{\mathcal{S}}'$ we have $(\nabla \cdot T )(h) = 0$
for all $h \in \mathcal{S}(\R^3;\R)$. This shows the claim.
\end{proof}

 As an immediate consequence of Theorem \ref{thm:mainschr1} we obtain the following lemma, which we shall use for the proof of  Theorem   \ref{thm:unitransformexp}.

\begin{lemma}  \label{lem:grp22} 
 Let  $\hh_B$ denote  the  real Hilbert space obtained by the completion of the inner product space  $( \underline{\mathcal{S}}_0 , B(\cdot, \cdot) )  $ 
with the imbedding $\iota :  \underline{\mathcal{S}}_0 \to \hh_B$ having dense range. Let  $v \in \hh_B$, and let $(v_n)_{n \in \N}$ be a    Cauchy sequence in $ \underline{\mathcal{S}}_0$ such that $\iota(v_n) \to v$.
Then  the following limit exists in  $L^2( \underline{\mathcal{S}}',d\nu)$
$$
\varphi(v) := \lim_{n \to \infty} \varphi(v_n) ,
$$
is independent of the Cauchy sequence. Furthermore,
 $\varphi(v)$ is a Gaussian random process indexed by $\hh_B$
with  $(\underline{\mathcal{S}}' ,\nu)$   the probability  measure space of the random process.
\end{lemma}
\begin{proof}
First observe that  $( \underline{\mathcal{S}}_0,B(\cdot,\cdot))$ is indeed an inner product space, since $\nabla \cdot  f = 0$ implies $P\widehat{f} = \widehat{f}$.
Clearly, $\varphi(v_n)$ is a Cauchy sequence  in $L^2( \underline{\mathcal{S}}',d\nu)$, since   $\int  |\varphi(v_n)-\varphi(v_m)|^2 d\nu = \frac{1}{2}B(v_n-v_m,v_n-v_m)$  by  Theorem \ref{thm:mainschr1} (a),
and hence converges to a unique limit. With regard  to Definition \ref{def:gaussrandproc}    the statement of the last sentence  is  straight forward to show using    Theorem \ref{thm:mainschr1} and the fact that limits of Gaussians are Gaussian.
\end{proof}

\begin{proof}[Proof of Theorem \ref{thm:unitransformexp}.]
The map  $i_\omega : ( \underline{\mathcal{S}}_0, B(\cdot,\cdot))  \to \{ v \in \vv : {\rm Im} v = 0 \}$ is an isometry of real inner product  spaces, which follows
directly from the definitions.
Furthermore, $i_\omega$  has dense range.
To see this,  observe that for any real  $v \in \vv$ there exists by well known construction  a real $v_n \in \underline{\mathcal{S}}$ such that $v_n \to v$ in the $L^2(\R^3;\C^3)$ norm.
Now  define $w_n =   ( \omega^{1/2}   (1-\chi_n)  P \hat{v}_n )^\vee$ for $\chi \in C_c^\infty(\R^3;[0,1])$ with $\chi = 1$ on $B_{1/2}(0)$
and $\chi =0$ outside of $B_1(0)$,
and $\chi_n(x) = \chi(nx)$.  Then it is straight forward to see that  $w_n \in \underline{\mathcal{S}}_0$ and (by unitarity of the Fourier transform and dominated
convergence)
\begin{align*}
& \| i_\omega w_n - v \| =  \| ( (1-\chi_n)  P \hat{v}_n)^\vee - v \| =  \| ( (1-\chi_n)  P \hat{v}_n - \hat{v} \| =  \| ( (1-\chi_n)  P \hat{v}_n -  P \hat{v} \|  \\
& \leq  \| \chi_n  P \hat{v} \| +  \| (1-\chi_n)  P ( \hat{v}_n  -  \hat{v} ) \|  \leq   \| \chi_n  P \hat{v} \| + \|  \hat{v}_n  -  \hat{v}  \| \to 0 ,
 \end{align*}
as $n \to \infty$, by construction. This  shows that $i_\omega$ has dense range.
So the map  $i_\omega$ extends to $\hh_B$  the closure of
$( \underline{\mathcal{S}}_0 , B(\cdot, \cdot) )  $
  and  yields a bijective isometry  $\hh_B \to \{ v \in \vv : {\rm Im} v = 0 \}$.
It follows using   Lemma   \ref{lem:grp22} that $\varphi \circ i_\omega^{-1}$ is a  Gaussian random
process indexed by $\{ v \in \vv : {\rm Im} v = 0 \}$ with probability measure space
$(\underline{\mathcal{S}}',\nu)$.  Thus it follows from Theorem \ref{thm:mainindex}
that there exists a unique unitary transformation $V_\vv : \FF_s(\vv)   \to L^2( \underline{\mathcal{S}}',d\nu)$
with
$V_\vv \Omega = 1$ and $V_\vv (\overline{  a^*(h) + a(h)}) V_\vv^{-1} = \varphi(h)$
 for all $h \in  i_\omega \underline{\mathcal{S}}_0$ (since $\{ v \in \vv : {\rm Im} v = 0 \}_\C = \vv$ and
$i_\omega \underline{\mathcal{S}}_0$ is dense in $\{ v \in \vv : {\rm Im} v = 0 \}$).
This shows the first part of the theorem. 
The last statement of the theorem now  follows form
part (b) of Theorem \ref{thm:mainindex}.
\end{proof}

%\bibliography{references}
%\bibliographystyle{amsplain}

\providecommand{\bysame}{\leavevmode\hbox to3em{\hrulefill}\thinspace}
\providecommand{\MR}{\relax\ifhmode\unskip\space\fi MR }
\providecommand{\MRhref}[2]{%
  \href{http://www.ams.org/mathscinet-getitem?mr=#1}{#2}
}
\providecommand{\href}[2]{#2}

\end{document}